\documentclass[a4 paper]{article}
\usepackage[utf8]{inputenc}
\usepackage[english]{babel}
\usepackage{amsthm}
\usepackage{amsmath}
\usepackage{amssymb}
\usepackage{mathtools}
\usepackage{multirow}
\usepackage{booktabs}
\usepackage{hyperref}
\hypersetup{hidelinks}
\usepackage{verbatim} 
\usepackage{slashed}
\usepackage{comment}
\usepackage{mathbbol}
\usepackage{verbatim}
\usepackage{ indentfirst}
\usepackage{microtype}
\usepackage[margin=1.1in]{geometry}
\usepackage{authblk}
\usepackage{mathrsfs}
\usepackage{environ}

\newcommand{\ie}{{\it i.e.}\ }

\usepackage{titlesec}
\titleformat{\chapter}[display]
  {\bfseries\Large}
  {\filright\MakeUppercase{\chaptertitlename} \Large\thechapter}
  {3.5ex}
  {\titlerule\vspace{1ex}\filleft}
  [\vspace{1ex}\titlerule]

\usepackage{caption}
\usepackage{subfig}

\numberwithin{equation}{section}

\NewEnviron{eqsplit}{%
\begin{equation}\begin{split}
  \BODY
\end{split}\end{equation}
}

\theoremstyle{plain}
\newtheorem{prp}{Proposition}[section]
\newtheorem{thrm}[prp]{Theorem}
\newtheorem{crll}[prp]{Corollary}
\newtheorem{lmm}[prp]{Lemma}

\theoremstyle{definition}
\newtheorem{dfn}[prp]{Definition}

\theoremstyle{remark}

\newcommand{\R}{\mathbb{R}}
\newcommand{\C}{\mathbb{C}}
\newcommand{\Ham}{\mathcal{H}}
\newcommand{\Lag}{\mathscr{L}}

\newcommand{\cpb}[2]{\{\! | #1, #2| \! \}}

\newcommand{\pb}[2]{\{ #1, #2 \}}
\newcommand{\ip}[2]{ #1\lrcorner #2 }

\newcommand{\omegone}{\Omega^{(1)}}
\renewcommand{\^}{\wedge}
\renewcommand{\epsilon}{\varepsilon}

\renewcommand{\tilde}{\widetilde}

\newcommand{\parder}[2]{\frac{\partial #1}{\partial #2}}
\newcommand{\varder}[2]{\frac{\delta #1}{\delta #2}}
\newcommand{\tilpartial}{\tilde \partial}

\newcounter{savefootnote}
\newcounter{symfootnote}
\newcommand{\symfootnote}[1]{%
   \setcounter{savefootnote}{\value{footnote}}%
   \setcounter{footnote}{\value{symfootnote}}%
   \ifnum\value{footnote}>8\setcounter{footnote}{0}\fi%
   \let\oldthefootnote=\thefootnote%
   \renewcommand{\thefootnote}{\fnsymbol{footnote}}%
   \footnote{#1}%
   \let\thefootnote=\oldthefootnote%
   \setcounter{symfootnote}{\value{footnote}}%
   \setcounter{footnote}{\value{savefootnote}}%
}

\begin{document}
\title{Hamiltonian multiform description of an integrable hierarchy}
\author{Vincent Caudrelier, Matteo Stoppato}

\begin{center}
	\textbf{\Large Hamiltonian multiform description of an integrable hierarchy}\\[3ex]
	\large{Vincent Caudrelier, Matteo Stoppato\symfootnote{Corresponding author's email address: \texttt{mmms@leeds.ac.uk}.}}\\[3ex]
	
	School of Mathematics, University of Leeds, LS2 9JT, UK  \\[3ex]
\end{center}

\vspace{0.5cm}

\begin{abstract}
Motivated by the notion of Lagrangian multiforms, which provide a Lagrangian formulation of integrability, and by results of the authors on the role of covariant Hamiltonian formalism for integrable field theories, we propose the notion of Hamiltonian multiforms for integrable $1+1$-dimensional field theories. They provide the Hamiltonian counterpart of Lagrangian multiforms and encapsulate in a single object an arbitrary number of flows within an integrable hierarchy. For a given hierarchy, taking a Lagrangian multiform as starting point, we provide a systematic construction of a Hamiltonian multiform based on a generalisation of techniques of covariant Hamiltonian field theory. This also produces two other important objects: a symplectic multiform and the related multi-time Poisson bracket. They reduce to a multisymplectic form and the related covariant Poisson bracket if we restrict our attention to a single flow in the hierarchy. Our framework offers an alternative approach to define and derive conservation laws for a hierarchy. We illustrate our results on three examples: the potential Korteweg-de Vries hierarchy, the sine-Gordon hierarchy (in light cone coordinates) and the Ablowitz-Kaup-Newell-Segur hierarchy.
\end{abstract}

\section{Introduction}

The objects and results presented in this paper, to be detailed below, come from the confluence of several new ideas that have emerged in the theory of integrable systems in recent years.
The first idea, introduced in 2009 by Lobb and Nijhoff \cite{Lobb_Nijhoff_2009} is the notion of \emph{Lagrangian multiforms}. The motivation was to address the completely open problem of characterising integrability of (partial) differential (or difference) equations purely from a variational/Lagrangian point of view. Despite the well known and fundamental interplay between Lagrangian and Hamiltonian formalism in classical and quantum physics, when it comes to integrable systems, one can only observe that the Hamiltonian approach has been the overwhelming favourite, mainly (but not fully) because of the extraordinary success of the canonical quantization procedure. This was carried out via the classical $r$-matrix approach \cite{Sklyanin_1979,SemenovTianShansky_1983} which leads to the quantum $R$-matrix approach \cite{Sklyanin_1979,Faddeev_Sklyanin_Takhtajan_1979}, both of which have given unifying frameworks for dealing with integrable systems and led to their own fully fledged research areas in (Poisson) geometry and quantum groups. 
Initially developed in the realm of fully discrete integrable systems, Lagrangian multiforms provide a framework whereby the notion of multidimensional consistency \cite{Nijhoff_2002,Bobenko_Suris_2002}, which captures the 
analog of the commutativity of Hamiltonian flows known in continuous integrable systems, is encapsulated in a generalised variational principle. The latter contains the standard Euler-Lagrange equations for the various equations forming an integrable hierarchy as well as additional equations, originally called {\it corner equations} which can be interpreted as determining the allowed integrable Lagrangians themselves. The set of all these equations is now called multiform Euler-Lagrange equations. The original work of Lobb and Nijhoff \cite{Lobb_Nijhoff_2009} stimulated a wealth of subsequent developments, first in the discrete realm, see e.g. \cite{Lobb_Nijhoff_Quispel_2009,Lobb_Nijhoff_2010,Bobenko_Suris_2010,Yoo-Kong_Lobb_Nijhoff_2011,Boll_Petrera_Suris_2014,Boll_Petrera_Suris_2015}, then progressively into the continuous realm for finite dimensional systems, see e.g. \cite{Suris_2013,Petrera_Suris_2017} and $1+1$-dimensional field theories, see e.g. \cite{Xenitidis_Nijhoff_Lobb_2011}, up to more recent developments in continuous field theory, see e.g. \cite{Suris_2016,Suris_Vermeeren_2016,Vermeeren_2019,Sleigh_Nijhoff_Caudrelier_2019,Petrera_Vermeeren_2019}, including the first example in $2+1$-dimensions \cite{Sleigh_Nijhoff_Caudrelier_2019_2}. 

Given that our focus is on $1+1$-dimensional field theories in this paper, let us present briefly the main ingredients of the theory of Lagrangian multiforms in this context. The starting point is to consider a two-form 
\begin{equation}
\label{lagrangian:multiform}
\Lag[u] = \sum_{i<j=1}^n L_{ij}[u] dx^{ij},~~n>2\,,
\end{equation} 
where for each $i,j$, $L_{ij}[u]$ is a function of a field\footnote{We only consider a single scalar field $u$ at this stage for simplicity of exposition but multicomponent fields are easily included. We will do so without further comment when we consider the example of the AKNS system in Section \ref{Section:AKNS}.} $u$ depending on the $n$ independent variables $x_1,\dots,x_n$ (the ``times'' of the hierarchy) and of the derivatives of $u$ with respect to these variables up to some finite order. We used the notation $dx^{ij}=dx^{i}\wedge dx^{j}$ and the convention $L_{ij}[u] =- L_{ji}[u]$. For convenience in this paper, we assume that the $L_{ij}[u]$ do not depend explicitly on the independent variables.
Associated to this two-form is an action 
\begin{equation}
S[u,\sigma] = \int_\sigma \Lag[u]\,,
\end{equation}
or rather a collection of actions, labelled by a $2$-dimensional (smooth) surface $\sigma$ in $\R^n$. At this stage, it is worthwhile noting that the 
standard variational approach to a field theory with two independent variables $x_1,x_2$ would consider a {\it volume form} $\Lag[u] =  L[u] dx^1\wedge dx^2$, with Lagrangian density $L[u]$, and simply an action $S[u] = \iint L[u]dx^1\wedge dx^2$. The novelty is in considering a $2$-form in a larger space\footnote{Note that we restrict this space to be of finite dimension $n$ here whereas strictly speaking, for an integrable field hierarchy one should let $n\to\infty$. The number $n$ corresponds to the number of commuting flows with respect to $x_1,\dots,x_n$ that we incorporate in the Lagrangian multiform. Our pragmatic approach is to consider $n$ fixed but arbitrary.} as well as an action labelled by a surface into this larger space.

The (generalised) equations of motion, called {\it multiform Euler-Lagrange equations}, are then obtained by postulating a (generalised) variational principle: one looks for critical points $u$ of $S[u,\sigma]$ simultaneously for all smooth surfaces $\sigma$ in $\R^n$ and, on critical points, the action is stationary with respect to arbitrary local variations of $\sigma$. The first requirement produces what is called multiform Euler-Lagrange equations which were given in \cite{Suris_Vermeeren_2016}. It can be shown \cite{Vermeeren_2018,Sleigh_Nijhoff_Caudrelier_2019_2} that they can be written compactly as $\delta d \Lag =0$ where we have used the two operator $\delta$ and $d$ arising in the variational bicomplex formalism (see below for a recap). The second requirement gives us the \emph{closure relation} on the equations of motion, i.e. the fact that on the equations of motion $d \Lag =0$. 

It turns out that this principle is much stronger than the standard variational principle because it not only produces certain Euler-Lagrange equations that determine the solutions for $u$, but also puts constraints on the Lagrangian coefficients $L_{ij}$ themselves. This was the original motivation for this principle which aims at capturing integrability in a variational fashion.  Needless to say that in general, it could well be that there are no solutions to the set of equations for the fields and for the Lagrangian multiform itself, or at least no nontrivial ones. We do not wish to include those cases so in the rest of this article, we assume that we have a nonzero form $\Lag$ whose multiform Euler-Lagrange equations are compatible and allow for (nontrivial) solutions. Such an $\Lag$ will be called a Lagrangian multiform if in addition it also satisfies the closure relation on the equations of motion. 
With this understanding, we take the following working definition:
\begin{dfn}
\label{def_Lag_multiform}
	The horizontal $2$-form \eqref{lagrangian:multiform} is a \emph{Lagrangian multiform} if $\delta d \Lag =0$ implies $d \Lag =0$.
\end{dfn}
Note that the requirement of imposing the closure relation on the equations of motion to define a Lagrangian multiform was an important feature of the original work of Lobb and Nijhoff. It has been dropped in some subsequent works by other authors and the related terminology is ``pluri-Lagrangian'' in that context. However, in the recent work \cite{Vermeeren_2020}, it is shown that this property is the Lagrangian counterpart of having Hamiltonian functions in involution. Our results here shed some more light on this connection between Lagrangian multiforms and Hamiltonians in involution, in the form of Theorem \ref{closure_Ham}. This clarifies the role of the closure relation to capture integrability in a Lagrangian framework.

The second idea, introduced by the authors in \cite{Caudrelier_Stoppato_2020} in the context of $1+1$-dimensional integrable field theories, is to use ideas from covariant Hamiltonian field theory, whose origins\footnote{The literature on this topic is vast and forms an entire community in its own right. It cannot be included here but we kindly refer the interested reader to the introduction of \cite{Caudrelier_Stoppato_2020} where an effort was made to point to key references, at least from the point of view of our work.} can be traced to the early work of de Donder and Weyl \cite{DeDonder_1930,Weyl_1935}, in conjunction with the $r$-matrix formalism. The latter had been confined to the standard Hamiltonian formalism since its introduction, thus breaking the natural symmetry between the independent space-time variables. The work \cite{Caudrelier_Stoppato_2020} had been motivated by earlier results \cite{Caudrelier_Kundu_2015,Caudrelier_2015,Avan_Caudrelier_Doikou_Kundu_2016} which showed a surprising spacetime duality in the classical $r$-matrix structure of a field theory. The origin of this duality was later explained in \cite{Avan_Caudrelier_2017} in the case of the Ablowitz-kaup-Newell-Segur (AKNS) hierarchy \cite{Ablowitz_Kaup_Newell_Segur_1974}. In \cite{Caudrelier_Stoppato_2020}, building on results and formalisms due for instance to Kanatchikov \cite{Kanatchikov_1998} and Dickey \cite{Dickey_2003}, we were able to construct a {\it covariant Poisson bracket} which possesses the classical $r$-matrix structure when evaluated on the natural Lax form of the theory. It is important to stress that the results were obtained from a single Lagrangian corresponding to the given $1+1$-dimensional theory at hand, \ie from a standard Lagrangian volume form. This allowed us to construct a multisymplectic form which in turn gave us access to the desired covariant Poisson bracket and $r$-matrix structure. We were also able to obtain the zero curvature representation, typical of integrable field theories, as a covariant Hamilton equation for the Lax form.

In view of this account, two natural questions arise:

\begin{enumerate}
	\item What happens to the construction of the multisymplectic form and the covariant Poisson bracket if we use a Lagrangian multiform instead of a standard Lagrangian (volume form) as a starting point?
	
	\item Provided the previous construction can be implemented, can we generalise the results of \cite{Caudrelier_Stoppato_2020} about the $r$-matrix structure of the covariant Poisson bracket to the structure that generalises this covariant Poisson bracket? In other words, can we extend the derivation of the $r$-matrix to a {\it multi-time Poisson bracket} that would appear when dealing with a whole integrable hierarchy?
\end{enumerate}

In the present paper, we investigate in detail the first question and leave the second question for our work \cite{Caudrelier_Stoppato_2020_3}. Specifically, in the context of $1+1$-dimensional integrable field theories, our main results are as follows:
\begin{itemize}
	\item We introduce a \textbf{Hamiltonian multiform} $\displaystyle \Ham = \sum_{i<j=1}^n H_{ij} dx^{ij}$, which is naturally associated to any given Lagrangian multiform. This requires to adapt techniques and notions from covariant Hamiltonian field theory and multisymplectic geometry that were conveniently cast into a purely algebraic framework in \cite{Dickey_2003}. We then prove the central result that $d\Ham=-2d\Lag$ on the multiform Euler-Lagrange equations. 
	 
	\item Alongside the Hamiltonian multiform associated to a Lagrangian multiform, our approach produces a generalisation of the multisympectic form that is canonically associated to a standard Lagrangian. We call it \textbf{symplectic multiform} for reasons that will be elaborated upon in the text. It naturally incorporates into a single object each symplectic form related to each Lagrangian $L_{1j}$ in the Lagrangian multiform. The symplectic multiform encapsulates the motion under the flows of with respect to the $n$ independent variables $x_1,\dots,x_n$ on our space of variables which forms the analog of the covariant phase space usually associated to first order field theories.
	
	\item Equipped with the symplectic multiform, we are able to define a \textbf{multi-time Poisson bracket}. It naturally incorporates certain {\it single time} Poisson brackets, which can be assembled naturally into pairs of dual Poisson brackets that were originally observed to possess the same $r$-matrix structure in \cite{Caudrelier_Kundu_2015,Caudrelier_2015,Avan_Caudrelier_Doikou_Kundu_2016}. Our multi-time Poisson bracket reproduces the covariant Poisson bracket in the case of $n=2$ independent (spacetime) variables, as considered for instance in \cite{Caudrelier_Stoppato_2020}. 
	
	\item We use these results to derive conservation laws traditionally signalling integrability in a field theory and whose construction has been the object of numerous studies based on fundamental ideas such as bi-Hamiltonian structures, recursion operators and Lax pairs, see e.g. \cite{Olver_2000}. Our method relies only on the elements introduced in our approach. This is implemented on the examples of the potential Korteweg-de Vries (pKdV) and AKNS Hamiltonian multiforms.
\end{itemize}

The paper is organised as follows. In Section \ref{Section:HamiltonianMultiform} we review the essential ingredients of the variational bicomplex as presented in \cite{Dickey_2003} in an algebraic language, instead of the original geometric approach, see e.g. \cite{Anderson_1989,Anderson_1992}. We use this framework to define and develop the theory of Hamiltonian multiforms, starting from a Lagrangian multiform. We then show how conservation laws fit into this context. In Sections \ref{Section:pKdV}, \ref{Section:SG} and \ref{Section:AKNS}, we illustrate the various constructions on the examples respectively of the pKdV hierarchy, of the sine-Gordon (sG) hierarchy in light-cone coordinates, and the AKNS hierarchy.

\section{Hamiltonian multiform, symplectic multiform and multi-time Poisson brackets}\label{Section:HamiltonianMultiform}
In this section, we first review essential notions and notations for our purposes in Sections \ref{varbi} and \ref{multi}. By presenting what is well known for a field theory associated to a Lagrangian volume form, {\it i.e.} covariant Hamiltonian, multisymplectic form and covariant Poisson bracket, it will be easier to appreciate the novelty brought in by the transition to the Lagrangian multiform framework, despite several of the defining relations looking superficially the same.  In particular, in analogy with the three fundamental objects just mentioned, we will introduce in Sections \ref{multi_Ham} and \ref{multiPB} the notion of Hamiltonian multiform, symplectic multiform and multi-time Poisson brackets together with their basic properties. 

\subsection{Elements of variational calculus with the variational bicomplex}\label{varbi}

The intuition behind the variational bicomplex formalism for a field theory can be summarised as follows. Let $M$ be the (spacetime) manifold with local coordinates $x^i$, $i=1, \ldots, n$. The manifold $M$ is viewed as the base manifold in a fibered manifold $\pi:E\to M$ whose sections represent the fields of the theory. The variational bicomplex is a double complex of differential forms defined on the infinite jet bundle of $\pi:E\to M$. One introduces vertical and horizontal differentials $\delta$ and $d$ which satisfy
\begin{equation}
d^2=0=\delta^2\,,~~d\delta=-\delta d\,,
\end{equation}
so that the operator $d+\delta$ satisfies $(d+\delta)^2=0$.

We now follow \cite{Dickey_2003} for a more detailed exposition of what we need in this paper. For convenience, we will only consider theories whose Lagrangian do not depend explicitly on the independent variables $x_i$. Let $\mathcal{K}=\R$ or $\C$. Consider the differential algebra with the commuting derivations $\partial_i$, $i=1, \dots, n$ generated by the commuting variables\footnote{Recall that for the most part in this paper, we will only consider a scalar field, $N=1$, except for the AKNS example.} $u_k^{(i)}$, $k=1, \dots, N$, $(i)=(i_1, \dots, i_n) $ being a multi-index, quotiented by the relations 
\begin{equation}
\partial_j u_k^{(i)}=u_k^{(i)+e_j}\,,
\end{equation}
where $e_j=(0,\dots,0,1,0,\dots,0)$ only has $1$ in position $j$. 
We simply denote $u_k^{(0,\dots,0)}$ by $u_k$, the fields of the theory which would be the local fibre coordinates mentioned above.
We denote this differential algebra by $\mathcal{A}$. We will need the notation 
\begin{equation}
\partial^{(i)}=\partial_1^{i_1}\partial_2^{i_2}\dots \partial_n^{i_n}\,.
\end{equation}
We consider the spaces $\mathcal{A}^{(p,q)}$, $p,q\ge 0$ of {\it finite} sums of the following form
\begin{equation}
\label{form_pq}
\omega =\sum_{(i),(k),(j)} f^{(i)}_{(k),(j)}\delta u_{k_1}^{(i_1)} \^ \dots \^ \delta u_{k_p}^{(i_p)} \^ dx^{j_1} \^ \dots \^ dx^{j_q}, \qquad f^{(i)}_{(k),(j)}\in \mathcal{A}
\end{equation}
which are called $(p,q)$-forms. In other words, $\mathcal{A}^{(p,q)}$ is the space linearly generated by the basis elements $\delta u_{k_1}^{(i_1)} \^ \dots \^ \delta u_{k_p}^{(i_p)} \^ dx^{j_1} \^ \dots \^ dx^{j_q}$ over $\mathcal{A}$, where $\wedge$ denotes the usual exterior product. We define the operations $d:\mathcal{A}^{(p,q)}\to \mathcal{A}^{(p,q+1)}$ and $\delta:\mathcal{A}^{(p,q)}\to \mathcal{A}^{(p+1,q)}$ as follows. They are graded derivations
\begin{subequations}
\begin{gather}
d(\omega_1^{(p_1,q_1)} \wedge \omega_2^{(p_2,q_2)}) = d\omega_1^{(p_1,q_1)} \wedge \omega_2^{(p_2,q_2)}+(-1)^{p_1+q_1}\omega_1^{(p_1,q_1)} \wedge d\omega_2^{(p_2,q_2)},\\
\delta(\omega_1^{(p_1,q_1)} \wedge \omega_2^{(p_2,q_2)}) = \delta\omega_1^{(p_1,q_1)} \wedge \omega_2^{(p_2,q_2)}+(-1)^{p_1+q_1}\omega_1^{(p_1,q_1)} \wedge \delta\omega_2^{(p_2,q_2)}\,,
\end{gather}
\end{subequations}
and on the generators, they satisfy
\begin{subequations}
\begin{gather}
df = \sum \partial_i \,f dx^i=\sum(\parder{f}{x^i} + \parder{f}{u_k^{(j)}} \, u_k^{(j)+e_i})dx^i\,, \quad f \in \mathcal{A}\,,\\
\delta f = \sum \parder{f}{u_k^{(i)}} \, \delta u_k^{(i)}\,, \quad f \in \mathcal{A}\,,\\ 
\delta(dx^i) = \delta(\delta u_k^{(j)}) = d(dx^i)=0,\\
d(\delta u_k^{(i)}) = - \delta d u_k^{(i)} = - \sum \delta u_k^{(i)+e_j} \^ dx^j.\label{ddeltagenerator}
\end{gather}
\end{subequations}
This determines the action of $d$ and $\delta$ on any form as in \eqref{form_pq}. As a consequence, one can show that $d^2=\delta^2=0$ and $d\delta =-\delta d$. For our purpose, it is sufficient to take the following (simplified) definition for the variational bicomplex: it is the space ${\cal A}^*=\bigoplus_{p,q}{\cal A}^{(p,q)}$ equipped with the two derivation $d$ and $\delta$. Due to the geometrical interpretation of these derivations, $d$ is called horizontal derivation while $\delta$ is called vertical derivation.
Note that the direct sum over $q$ is finite and runs from $0$ (scalars) to $n$ (volume horizontal forms) whereas the sum over $p$ runs from $0$ to infinity. Of course, each form in ${\cal A}^*$ only contains a finite sum of elements of the form \eqref{form_pq} for certain values of $p$ and $q$. The bicomplex ${\cal A}^*$ generates an associated complex ${\cal A}^{(r)}=\bigoplus_{p+q=r}{\cal A}^{(p,q)}$ and derivation $d+\delta$. It is proved that both the horizontal sequence and the vertical sequence are exact, see e.g. \cite{Dickey_2003}.\\
Dual to the notion of forms is the notion of vector fields. We consider the dual space of vector fields ${\cal T}{\cal A}$ to the space of one-forms ${\cal A}^{(1)}$ with elements of the form
\begin{equation}
\label{vector_field}
\xi=\sum_{k,(i)} \xi_{k,(i)}\, \partial_{u_{k}^{(i)}}+\sum_i \xi^*_i\, \partial_i\,.
\end{equation}
The interior product with a form is obtained in the usual graded way together with the rule
\begin{equation}
\ip{\partial_i}{dx^j}=\delta_{ij}\,,~~\ip{\partial_{u_{k}^{(i)}}}{\delta u_l^{(j)}}=\delta_{kl}\delta_{(i)(j)}\,.
\end{equation}
where $\delta_{(i)(j)} = \prod_k \delta_{i_k j_k}$. For instance, with $i\neq j$ and $(i)\neq (j)$ or $k\neq l$,
\begin{subequations}
\begin{gather}
\ip{\partial_i}{(\delta{u_k^{(l)}}\wedge dx^i \wedge dx^j)}=-\delta{u_k^{(l)}} \wedge dx^j\,,\\
\ip{\partial_{u_k^{(i)}}}{(\delta{u_l^{(j)}}\wedge \delta{u_k^{(i)}} \wedge dx^m)}=-\delta{u_l^{(j)}} \wedge dx^m\,.
\end{gather}
\end{subequations}
In particular, we will need the following vertical vector fields
\begin{equation}
\tilpartial_i = \sum_{k, (j)} u_k^{(j) + e_i} \parder{}{u_k^{(j)}}.
\end{equation}
Let us also introduce the notation $\partial'_i$ by $\partial_i=\partial'_i+\tilpartial_i$.
If $f\in \mathcal{A}$ does not depend explicitly on the space-time variables then $\partial_i f = \tilpartial_i f$. In addition to the vector fields \eqref{vector_field}, in general calculations in the variational bicomplex also require the use of multivector fields of the form $\xi_1\wedge \dots \wedge \xi_r$ where each $\xi_i$ is of the form \eqref{vector_field}. In this paper, we will mostly need those multivector fields that are linear combination of $\partial_{u_{k}^{(i)}}\wedge \partial_j$ with coefficients in ${\cal A}$ and we may simply call them vector fields as the context should not lead to any confusion. The following example shows the rule for the interior product of such a multivector field, with $(i)\neq (j)$ or $k\neq l$,
\begin{equation*}
\ip{(\partial_{u_{k}^{(i)}}\wedge \partial_\ell)}{(\delta{u_l^{(j)}}\wedge \delta{u_k^{(i)}} \wedge dx^m)}=\ip{\partial_{u_{k}^{(i)}}}\ip{( \partial_\ell}{(\delta{u_l^{(j)}}\wedge \delta{u_k^{(i)}} \wedge dx^m))}=-\delta_{\ell m}\,\delta{u_l^{(j)}}\,.
\end{equation*}
Finally, we will need the following useful identity, cf  \cite[Corollary 19.2.11]{Dickey_2003}.
\begin{equation}
\tilpartial_i = \delta \ip{\tilpartial_i} + \ip{\tilpartial_i}{\delta}\,.
\end{equation}

\subsection{The multisymplectic approach to a PDE}\label{multi}

Equipped with the above basic elements of the variational bicomplex, we now recall how to describe a partial differential equation admitting a Lagrangian formulation into a covariant Hamiltonian formulation. This serves as a basis to introduce known results and objects, in particular the multisymplectic form. What is reviewed here will be helpful to identify the novel ingredients in the rest of this paper. 

Recall that we focus on two-dimensional field theories so our starting point is a \emph{Lagrangian volume $2$-form}
\begin{equation}
\Lambda = L \,dx^1\wedge dx^2.
\end{equation}
$L$ is the Lagrangian density and depends on the fields $u^k$ $k=1,\dots,N$ and their derivatives with respect to $x^1$ and $x^2$. We use the variational bicomplex described in the previous section with $n=2$. It is known that there exist unique elements $A_k\in\mathcal{A}$ such that 
\begin{equation}
\label{def_omegaone}
\delta \Lambda = \sum_k A_k\,\delta u^k \wedge dx^1 \wedge dx^2 - d \omegone
\end{equation}
where  $\omegone\in{\cal A}^{(1,1)}/d{\cal A}^{(1,0)}$ is only determined up to a total horizontal derivative. The coefficients $A_k$ are denoted 
 $\varder{L}{u^k}$ and are the variational derivatives with respect to $u_k$. One then obtains the Euler-Lagrange equations by setting $\varder{L}{u^k}=0$ for every $k$. $\omegone$ is obtained using the property $\delta d + d \delta=0$ in $\delta \Lambda$ as much as possible. Of course, the content of this result is simply the local analog of the standard integration by parts procedure used when varying the action $\int \Lag$. In the latter, the boundary term $\int d \omegone$ is usually discarded. To understand the role played by $\omegone$, we remark the following facts. For a classical finite-dimensional Lagrangian system, this is (the pull-back to the tangent bundle of) the canonical one form $\parder{L}{\dot q} \delta q$, and one can obtain the symplectic form by taking its $\delta$-differential. Similarly, in the case of field theories where $\Lag$ is taken to be a volume form, the form is $\omegone = \omega^{(1)}_1 \wedge dx^1 + \omega^{(1)}_2 \wedge dx^2$ where $\omega^{(1)}_1$ and $\omega^{(1)}_2$ each have a similar structure to the canonical one form of the finite dimensional case. It contains the {\it usual} symplectic structure $-\omega^{(1)}_1$ (if we consider $x_2$ as our {\it time}) but also the dual structure $\omega^{(1)}_2$ (which would correspond to performing the Legendre transform when choosing $x_1$ as the time variable). To summarize, for a field theory, $\omegone$ realises the Legendre transform simultanously with respect to all independent variables.
 
The next step is to define the \emph{covariant Hamiltonian} as
\begin{equation}
\label{def_cov_Ham}
H = - \Lambda + \sum_{j=1}^2 dx^j \wedge \ip{\tilpartial_j }{\omegone}.
\end{equation}
and the \emph{multisymplectic form} $\Omega\in \mathcal{A}^{(2,1)}$ as 
\begin{equation}
\label{def_multisymplectic_form}
    \Omega = \delta \omegone\,. 
\end{equation}
One obtains the covariant Hamilton equations as 
\begin{equation}
\delta H = \sum_{j=1}^2 dx^j \wedge \ip{\tilpartial_j}{\Omega},
\end{equation}
which are equivalent to the to the Euler-Lagrange equation, as they should.
In general, let us note that if a PDE involves $n$ independent variables and admits a Lagrangian description, $\Lambda$ and $H$ are volume $n$-forms, $\omegone \in \mathcal{A}^{(1,n-1)}$ and $\Omega \in \mathcal{A}^{(2,n-1)}.$\\
 
Equipped with a multisymplectic form we can consider the definition of a covariant Poisson bracket, following for instance Kanatchikov \cite{Kanatchikov_1998}. We stress that the definition of a covariant Poisson bracket from a multisymplectic form, in a way that mimics the situation in classical mechanics, has been part of a rich activity since the early proposals. In particular, the Jacobi identity is a delicate issue, as well as the need to restrict to certain forms, called Hamiltonian, as we explain below. We refer the reader to \cite{Forger_Salles_2015} for a detailed account. For our purpose, we will simply use Kanatchikov's ideas and adapt them to our purposes. The results of \cite{Caudrelier_Stoppato_2020} show that, at least in our context, this leads to a satisfactory covariant Poisson bracket satisfying the Jacobi identity, thanks to the fact that the latter translates into the classical Yang-Baxter equation for the classical $r$-matrix.

We need to restrict our attention to the a special class of forms called Hamiltonian. We take the following definition which is sufficient for our purposes: a horizontal form $F$ is said to be Hamiltonian if there exists a (multi)vector field $\xi_F$ such that $\ip{\xi_F}{\Omega} = \delta F$. Contrary to the usual symplectic case, the property of Hamiltonianicity is quite restrictive in the multisymplectic case.

For two Hamiltonian forms $P$ and $Q$, of (horizontal) degree respectively $r$ and $s$, we can define their \emph{covariant Poisson bracket} as
\begin{equation}
\label{def_cov_PB}
\cpb{P}{Q}_c = (-1)^{r} \ip{\xi_P}\ip{\xi_Q}\Omega.
\end{equation}
We now state the following fact, which was only obtained explicitely on examples in \cite{Caudrelier_Stoppato_2020}, but for which no general proof was given. 
\begin{prp}
If the covariant Hamiltonian density $h = \ast^{-1} H$ is a Hamiltonian form, then we have for any Hamiltonian $1$-form $F$
\begin{equation}
d F = \cpb{h}{F}_c dx^1\wedge dx^2.
\end{equation}
\end{prp}
This is of course the multisymplectic analog of the well-known equation in Hamiltonian mechanics $dF=\{H,F\}dt$ giving the time evolution of a function $F$ on the phase space under the Hamiltonian flow of $H$. 
In this paper, we will give this statement and a proof in the more general setting of Section \ref{multiPB}, from which the above can be recovered by setting $n=2$.

\subsection{The multiform Hamilton equations and the symplectic multiform}\label{multi_Ham}

The  main observation at the basis of this paper is that the objects and results reviewed in Section \ref{multi} can be extended to a Lagrangian multiform, \ie a horizontal $2$-form
\begin{equation}
\Lag[u] = \sum_{i<j}^n L_{ij}[u]\, dx^{ij}, \qquad n > 2\,,
\end{equation}
required to satisfy a generalised variational principle associated to the action 
\begin{eqnarray}
S[u,\sigma]=\int_\sigma \Lag[u]\,,
\end{eqnarray}
as explained in the introduction.
We can now turn our attention to the generalisation of the form $\omegone$ in \eqref{def_omegaone}. 
We first use the following result from \cite[Proposition 6.3]{Vermeeren_2018} and \cite{Vermeeren_2020}, which we reproduce here with a little change of notation.
\begin{prp}\label{omega1:prop}
The field $u$ is a critical point of $S[u,\sigma]=\int_\sigma \Lag[u]$ for all (smooth) surface $\sigma$ in $\mathbb{R}^n$ if and only if there exists a (nonzero) form $\omegone$ such that
\begin{equation}
\delta \Lag[u] = - d \omegone\,.
\end{equation}
\end{prp}
We also recall that, as explained in the introduction we have that $u$ is a critical point of $S[u,\sigma]$ for all (smooth) surface $\sigma$ if and only if $\delta d\Lag=0$. Equipped with this, let us write (dropping the dependence on $u$ for conciseness),
\begin{eqnarray}
{\cal E}(\Lag):=\delta \Lag + d \omegone\,.
\end{eqnarray}
Then, a reformulation of the previous discussion is as follows:
\begin{eqnarray}
\label{reformulation}
\delta d\Lag=0\Leftrightarrow \text{$u$ is a critical point of $S[u,\sigma]$ for all (smooth) surface $\sigma$ in $\mathbb{R}^n$} \Leftrightarrow{\cal E}(\Lag)=0\,.
\end{eqnarray}
Compared to the case of \eqref{def_omegaone}, in addition to the nonuniqueness of $\omegone$ induced by the freedom of adding a total differential $d\omega$ to $\Lag$ (as for a standard Lagrangian volume form), there is also some freedom in the integration by parts steps which lead to the expression 
\begin{eqnarray}
\label{decomp1}
\delta \Lag={\cal E}(\Lag)- d \omegone\,.
\end{eqnarray}
More precisely, in general we could also have another way of writing $\delta \Lag$,
\begin{eqnarray}
\label{decomp2}
\delta \Lag=\widetilde{\cal E}(\Lag)- d \widetilde\omegone\,,
\end{eqnarray}
with still $\widetilde{\cal E}(\Lag)=0 \Leftrightarrow \delta d\Lag=0$, following from Proposition \ref{omega1:prop} and reformulation \eqref{reformulation}. We will show that these two sources of freedom have no consequence on our constructions. 
Equipped with a pair $(\Lag,\omegone)$, we define the Hamiltonian multiform associated to it.
\begin{dfn}[Hamiltonian multiform]
The \textbf{Hamiltonian multiform} associated to the pair $(\Lag,\omegone)$ is defined by
\begin{equation}
\Ham = - \Lag + \sum_{j=1}^n dx^j \wedge \ip{\tilpartial_j }{\omegone}.
\end{equation}
\end{dfn}
As announced, this definition looks very similar to the definition of the covariant Hamiltonian in \eqref{def_cov_Ham}. However note that the sum involves $n$ terms here (the number of independent variables included in the Lagrangian multiform) and that $\Ham$ has the form $\displaystyle \Ham = \sum_{i<j} H_{ij} dx^{ij}$ and is in ${\cal A}^{(0,2)}$, like $\Lag$. $\Ham$ plays the role of the covariant Hamiltonian form in the multiform context. 
\begin{lmm}
\label{total_diff}
Let $\widetilde{\Lag}=\Lag+d\omega$ for some $\omega\in {\cal A}^{(0,1)}$, let $\Ham$ be the Hamiltonian multiform associated to the pair $(\Lag,\omegone)$ and $\widetilde{\Ham}$ the one associated to the pair $(\widetilde{\Lag},\omegone+\delta\omega)$. Then,
\begin{equation}
    \widetilde{\Ham}=\Ham\,.
\end{equation}
\end{lmm}
\begin{proof}
We have 
\begin{eqnarray}
\widetilde{\Ham}&=& - \Lag -d\omega + \sum_{j=1}^n dx^j \wedge \ip{\tilpartial_j }{(\omegone+\delta \omega)} \nonumber\\
&=&\Ham-d\omega + \sum_{j=1}^n dx^j \wedge \tilpartial_j \omega-\sum_{j=1}^n dx^j \wedge \delta \ip{\tilpartial_j }{\omega}\,,
\end{eqnarray}
where we have used the property $\tilpartial_j = \delta \ip{\tilpartial_j} + \ip{\tilpartial_j} \delta$. Now since $\omega\in{\cal A}^{(0,1)}$, $\ip{\tilpartial_j }{\omega}=0$ for $j=1,\dots,n$. Recall that we work with Lagrangians that do not depend explicitly on the space-time variables, hence neither does $\omega$ so $\displaystyle \sum_{j=1}^n dx^j \wedge \tilpartial_j \omega=d\omega$ and the result follows.
\end{proof}
The relevance of this lemma is related to the symplectic multiform defined below and the multiform Hamilton equations associated to it and $\Ham$.

We can easily see that there is a relation between the $d$-differential of $\Ham$ and the one of $\Lag$. The next result is important and connects the closure relation in the Lagrangian multiform to the Hamiltonian multiform formalism.
\begin{thrm}\label{closure_Ham}
$d \Ham = - 2 d \Lag$ modulo the multiform E-L equations.
\end{thrm}
\begin{proof}
We start from the definition of $\Ham$:
\begin{equation}
d \Ham = - d \Lag + d \left(  \sum_{j=1}^n dx^j \wedge \ip {\tilpartial_j }\omegone \right) = - d \Lag - \sum_{j=1}^n dx^j \wedge d \ip {\tilpartial_j }\omegone = - d \Lag +  \sum_{j=1}^n dx^j \wedge \ip{ \tilpartial_j }d \omegone
\end{equation}
where we used $d \ip{ \tilpartial_j }+ \ip {\tilpartial_j }d = 0$. Now we use the equation $\delta \Lag = - d \omegone$ to obtain
\begin{equation}
\begin{gathered}
d \Ham = - d \Lag - \sum_{j=1}^n dx^j \wedge \ip {\tilpartial_j }\delta \Lag= - d \Lag - \sum_{j=1}^n dx^j \wedge( \tilpartial_j - \delta \ip{ \tilpartial_j} ) \Lag\\ = - d \Lag -  \sum_{j=1}^n dx^j \wedge \tilpartial_j  \Lag = - 2 d \Lag.
\end{gathered}
\end{equation} 
In the last line we used the property $\tilpartial_j = \delta \ip{\tilpartial_j} + \ip{\tilpartial_j} \delta$, and the fact that $\Lag$ is purely horizontal and does not depend explicitly on the space-time variables. 
\end{proof}
We remark that in \cite{Suris_Vermeeren_2016,Vermeeren_2020} the closure of a pluri-Lagrangian form $\Lag$ was linked to the involution of the single-time Hamiltonians, and in \cite{Vermeeren_2020} an analogue of Theorem \ref{closure_Ham} for the case of Lagrangian 1-forms was given. In the particular case where the Hamiltonian multiform is a Hamiltonian form in the sense defined below, we expect Theorem \ref{closure_Ham} to provide a general framework in which to recast these results (with appropriate modifications for the examples in $0+1$ dimensions presented in \cite{Suris_Vermeeren_2016,Vermeeren_2020}). This point is left for future investigation. Recalling that a Lagrangian multiform is defined to satisfy the closure relation on the equations of motion, we obtain:
\begin{crll}[Closedness of $\Ham$]
The Hamiltonian multiform is horizontally closed on the multiform E-L equations $d \Ham =0$. In other words, $\Ham$ satisfies the closure relation.
\end{crll}
These results justify our terminology {\it Hamiltonian multiform} since we have the closure relation for $\Ham$ if and only if it holds for $\Lag$.
This corollary is the multiform equivalent of the well known fact in finite-dimensional mechanics that the Hamiltonian is a conserved quantity $\frac{dH}{dt}=0$ (recall that we do not include explicit dependence on the independent variables here).\\

We are now in a position to introduce the multiform analog of the multisymplectic form \eqref{def_multisymplectic_form}, again denoting it by $\Omega$.
\begin{dfn}
The \textbf{symplectic multiform} associated to $\omegone$ is defined as $\Omega = \delta \omegone$.
\end{dfn}
It is clear from Lemma \ref{total_diff} that adding a total differential $d\omega$ to $\Lag$, which amounts to adding $\delta \omega$ to $\omegone$, has no consequence on $\Omega$.

The reader will hopefully forgive us for the choice of terminology, very similar to multisymplectic form. Another candidate, polysymplectic form, is already in use in the literature. We could not simply keep 
multisymplectic form for our new object since, although both objects are derived in a similar fashion and play a similar role in the theory, they are quite different in structure. Indeed, recall that the multisymplectic form for a theory with $n$ independent variables would be in ${\cal A}^{(2,n-1)}$ whereas our symplectic multiform is in ${\cal A}^{(2,1)}$ so they only coincide in the case where $n=2$ (for our case of $1+1$ field theories). The symplectic multiform is of the form
\begin{equation}\label{omega:decomp}
\Omega = \sum_{j=1}^n \omega_j \wedge dx^j, \qquad \omega_j \in \mathcal{A}^{(2,0)}\,,~~j=1,\dots,n.
\end{equation}
The following corollary gives support for our terminology as it is reminiscent of the fact that a symplectic form $\omega$ is closed in classical mechanics. 
\begin{crll}
The symplectic multiform is horizontally closed on the multiform E-L equations
\begin{equation}
\delta d \Lag = 0 \qquad \implies \qquad d \Omega =0.
\end{equation}
\end{crll}
\begin{proof}
The equations are equivalent to $\delta \Lag = -d \omegone $, so
\begin{equation}
0=\delta^2 \Lag = -\delta d \omegone = d \delta \omegone = d \Omega. \qedhere
\end{equation}
\end{proof}
The symplectic multiform $\Omega$ achieves an important unification of the various (standard and dual) symplectic structures appearing in an integrable hierarchy, as originally observed in \cite{Avan_Caudrelier_Doikou_Kundu_2016}. When $x_1$ is chosen to be the $x$ variable and $x_j$, $j\ge 2$ to be the higher times $t_j$ of the hierarchy then $\omega_1$ represents (up to a sign) the usual symplectic form, while each $\omega_j$, $j \ne 1$ represents the dual symplectic form related to the time $t_j$. For each $j\ge 2$, the multisymplectic form $\Omega_{1j}$ which would be obtained by considering the Lagrangian $L_{1j}$ as a standalone Lagrangian, as in Section \ref{multi}, is simply obtained by taking $\omega_1 \wedge dx^1 + \omega_j\wedge dx^j$. 

We now use the symplectic multiform to obtain the {\it multiform Hamilton equations}.
\begin{prp}[multiform Hamilton equations]
The multiform Euler-Lagrange equations for the Lagrangian multiform $\Lag$ are equivalent to
\begin{equation}\label{Hamilton:equations}
\delta \Ham = \sum_{j=1}^n dx^j \wedge \ip{\tilpartial_j}{\Omega}.
\end{equation}
\end{prp}
\begin{proof}
The proof is a simple adaptation of the similar result obtained in \cite[Chapter 19]{Dickey_2003} to the multiform case. From the definition of $\Ham$ we get
\begin{equation}
\delta \Ham = - \delta \Lag - \sum_{j=1}^n dx^j \wedge \delta  \ip{\tilpartial_j}\omegone.
\end{equation}
Thanks to Proposition \ref{omega1:prop} the equations of motion are equivalent to 
\begin{equation}
\delta \Ham = d\omegone - \sum_{j=1}^n dx^j \wedge \delta  \ip{\tilpartial_j}\omegone = \sum_{j=1}^n dx^j \wedge (\partial'_j+\ip{\tilpartial_j} \delta + \delta \ip{\tilpartial_j})\omegone- \sum_{j=1}^n dx^j \wedge \delta  \ip{\tilpartial_j}\omegone.
\end{equation}
$\omegone$ does not depend explicitly on the space-time variables so $\partial'_j \omegone =0$. The result is obtained by cancellation. Lemma \ref{total_diff} ensures that the freedom of adding a total differential to $\Lag$ has no consequence on the multiform Hamilton equations as it should. The other source of freedom coming from \eqref{decomp1}-\eqref{decomp2} does not affect the result either. Indeed, suppose that $\widetilde{\Ham}$ is the Hamiltonian multiform associated to the pair $(\Lag,\widetilde{\omegone})$ of \eqref{decomp2} and $\widetilde{\Omega}$ is associated to $\widetilde{\omegone}$ then exactly the same computation as above yields that the multiform Euler-Lagrange equations for the Lagrangian multiform $\Lag$ are equivalent to 
\begin{equation}
\delta \widetilde{\Ham} = \sum_{j=1}^n dx^j \wedge \ip{\tilpartial_j}{\widetilde{\Omega}}.
\end{equation}
\end{proof}

\subsection{Multi-time Poisson brackets and conservation laws}\label{multiPB}

Continuing with the inspiration given by covariant Hamiltonian field, the next step is to try to construct a Poisson bracket related to our symplectic multiform and investigate how the multiform Hamilton equations can be cast into Poisson Bracket form. Similarly to the situation reviewed at the end of Section \ref{multi}, this can only be done for a restricted class of forms, called Hamiltonian forms. For convenience, we restrict our attention to horizontal forms as this is sufficient for our purposes.
\begin{dfn}[Hamiltonian forms]
We will say that a horizontal form $P$ is \emph{Hamiltonian} if there exists a (multi)vector field $\xi_P$ such that 
\begin{equation}
\xi_P \lrcorner \Omega = \delta P.
\end{equation}
$\xi_P$ is called the \emph{Hamiltonian vector field} related to $P$.
\end{dfn}
\begin{prp}
$P$ can be a non-trivial Hamiltonian form only if either $P \in \mathcal{A}$ or $P \in \mathcal{A}^{(0,1)}$.
\end{prp}
\begin{proof}
The proof follows from a simple counting argument. Suppose $P \in \mathcal{A}^{(0,s)}$. Then, since $\Omega \in \mathcal{A}^{(2,1)}$, in order for a $(p,q)$-vector field $\xi_P$ to exist such that
\begin{equation}
\ip{\xi_P}{\Omega} = \delta P
\end{equation}
then necessarily $2-p = 1$ and $1-q = s$. So $p=1$ and $q = 1 -s \ge 0$, and therefore $s$ can only be 0 or 1.
\end{proof}
 We now produce a statement that is similar to \cite[Proposition 2]{Caudrelier_Stoppato_2020}, but for the multiform case. The proof is easily obtained as an extension. We will use this result systematically without quoting it in our examples below. 

 Let us denote by $S_\Omega$ the set of basis elements $\delta u_l^{(i)}$ that appear explicitly the symplectic multiform. It is a finite set since $\Omega$ is derived from $\Lag$ which is assumed to depend on $u_l^{(i)}$ with $|i|\le m$ for some $m$ (finite jet dependence). Hence, we can assume some ordering on $S_\Omega$ such that we can label the $\delta u_l^{(i)}$'s as $\delta v_j$, $j=1,\dots,|S_\Omega|$. We then write
\begin{equation}\label{smform:localcoordinates}
\Omega=\sum_{k=1}^n \sum_{\substack{i<j\\i,j \in I_k}}\omega_k^{ij} \delta v_i \wedge \delta v_j \wedge dx^k
\end{equation}
where $I_k \subseteq \{1,\dots,|S_\Omega|\}$ for each $k=1,\dots,n$. Note that each $\omega_{k}^{ij}\in{\cal A}$ so has a dependence on the local coordinates $u_m^{(j)}$ which we do not show explicitly.
\begin{prp}[Necessary form of a Hamiltonian one-form.]
Suppose $\displaystyle F= \sum_{k=1}^n F_k\, dx^k$, where $F_{k}\in{\cal A}$, is a Hamiltonian form for the multisymplectic form \eqref{smform:localcoordinates}. Then, for each $k=1,\dots,n$, $F_k$ can only depend (at most) on $v_j$, $j \in I_k$.
\end{prp}
We can now define the multi-time Poisson brackets for Hamiltonian forms, in analogy with the covariant Poisson bracket. 
\begin{dfn}[multi-time Poisson brackets]
For two Hamiltonian forms $P$ and $Q$, of degree respectively $r$ and $s$, we define their multi-time Poisson bracket as
\begin{equation}
\cpb{P}{Q} = (-1)^{r} \xi_P \lrcorner \delta Q.
\end{equation}
\end{dfn}
This definition is formally the same as the one given by Kanatchikov, cf \eqref{def_cov_PB}, but we stress that since the degree of the symplectic multiform is $(2,1)$ (for every $n$) is different from the degree of the multisymplectic form, which is $(2,n-1)$ in general, then the resulting degree of the Poisson bracket of two horizontal forms will be different. In particular, we see that the multi-time Poisson bracket of two horizontal 1-forms is still a horizontal 1-form. The two brackets coincide when $n=2$. These Poisson brackets are graded antisymmetric and bilinear in the space of Hamiltonian forms. In particular
\begin{itemize}
\item $P,Q \in \mathcal{A}^{(0,1)}$, then $\cpb{P}{Q} = - \ip{\xi_P}{\delta Q} =  - \cpb{Q}{P} = \ip{\xi_Q}{\delta P}$;
\item $P \in \mathcal{A}^{(0,1)}$ and $H \in \mathcal{A}$, then $\cpb{H}{P} = \ip{\xi_H}{\delta P} = -\cpb{P}{H} = \ip{ \xi_P}{\delta H}$.
\end{itemize}
As mentioned before for the covariant Poisson bracket, our definition may lead to issues regarding the Jacobi identity for instance. However, in the spirit of \cite{Caudrelier_Stoppato_2020}, we investigate this further in \cite{Caudrelier_Stoppato_2020_3} in connection with the $r$-matrix structure of the multi-time Poisson bracket whereby the Jacobi identity translates into the classical Yang-Baxter equation. 
\begin{thrm}\label{df:hamiltonequations}
On the equations of motion
\begin{equation}
dF = \ip{\xi_F}{\delta \Ham}
\end{equation}
for any Hamiltonian $1$-form that does not depend on the independent variables.
\end{thrm}
\begin{proof}
Using \eqref{Hamilton:equations} and the antisymmetry of $\Omega$ we have
\begin{equation}
 \ip{\xi_F}{\delta \Ham} = \ip{\xi_F} \sum_{j=1}^n dx^j \wedge \ip{\tilpartial_j}\Omega = - \sum_{j=1}^n dx^j \wedge \ip{\xi_F}\ip{\tilpartial_j} \Omega = \sum_{j=1}^n dx^j \wedge \ip{\tilpartial_j}\ip{\xi_F} \Omega.
\end{equation}
Since $\ip{\xi_F}\Omega = \delta F$ we obtain
\begin{equation}
 \ip{\xi_F}{\delta \Ham} =\sum_{j=1}^n dx^j \wedge \ip{\tilpartial_j} \delta F.
\end{equation}
Using the property $\ip{\tilpartial_j}\delta = \tilpartial_j - \delta \ip{\tilpartial_j}$
\begin{equation}
 \ip{\xi_F}{\delta \Ham} =\sum_{j=1}^n dx^j \wedge \tilpartial_j F -\sum_{j=1}^n dx^j \wedge  \delta \ip{\tilpartial_j} F.
\end{equation}
Since $F$ is purely horizontal $\ip{\tilpartial_j} F=0$, and since it does not depend explicitly on the space-time variables $\tilpartial_jF = \partial_jF$, so that
\begin{equation}
 \ip{\xi_F}{\delta \Ham} =\sum_{j=1}^n dx^j \wedge \partial_j F = dF. \qedhere
\end{equation}
\end{proof}
If the components $H_{ij}$ of $\Ham$ are Hamiltonian $0$-forms, then the previous proposition leads to:
\begin{crll}
On the equations of motion
\begin{equation}
dF =  \sum_{i<j=1}^n \cpb{H_{ij}}{F} dx^{ij}.
\end{equation}
for any Hamiltonian $1$-form that does not depend on the independent variables.
\end{crll}
\begin{proof}
$$dF =\ip{\xi_F}{\delta \Ham} =\sum_{i<j=1}^n \ip{\xi_F}\delta H_{ij} \wedge dx^{ij}= - \sum_{i<j=1}^n \cpb{F}{H_{ij}} dx^{ij} =  \sum_{i<j=1}^n \cpb{H_{ij}}{F} dx^{ij}\,.$$
\end{proof}
This is a generalisation of the usual Hamilton equations in Poisson Bracket form for classical finite-dimensional mechanics $\dot f = \{H,f\}$. In our context, this result turns out to be useful in relation to conservation laws within an integrable hierarchy. Indeed, if $F$ is a $1$-form, we have
\begin{equation}
dF = \sum_{j=1}^n dx^j \wedge \partial_j F = \sum_{i,j=1}^n \partial_i F_j dx^i \wedge dx^j =  \sum_{i<j=1}^n (\partial_i F_j - \partial_j F_i) dx^i \wedge dx^j
\end{equation}
which means that, in fact if $dF=0$ on the equations of motion, then 
\begin{equation}
\partial_i F_j = \partial_j F_i, \qquad \forall i \ne j.
\end{equation}
This suggests the following
\begin{dfn}\label{cons_law}
We say that a Hamiltonian 1-form $F$ is a \emph{conservation law} if $dF=0$ on the equations of motion.
\end{dfn}
It is then immediate from Proposition \ref{df:hamiltonequations} that
\begin{crll}
A Hamiltonian 1-form $F$ is a conservation law if and only if on the equations of motion
\begin{equation}
\ip{\xi_F}{\delta \Ham}=0\,.
\end{equation}
\end{crll}
This is clearly an extension of the concept of first integral in classical mechanics. As we will show on some examples below, the very definition of a Hamiltonian form being a conservation law can lead to its explicit form. This is a rather elegant byproduct of our approach.\\

We now address the relationship between the multi-time Poisson bracket that we just defined and the single-time Poisson brackets that can be derived from the single Lagrangians $L_{ij}$ using the usual construction.
Starting from the decomposition \eqref{omega:decomp}, for each $i=1,\dots,n$, it is natural to want to define the $i$-th Poisson bracket of two $0$-forms $f,g\in{\cal A}$ as
\begin{equation}
\{f,g\}_i := - \ip{\xi^i_f}{\delta g},\qquad \mbox{where }\qquad \ip{\xi^i_f} {\omega_i} = \delta f.  
\end{equation} 
We remark that there is no sum on the $i$ index. Compared to the standard finite-dimensional case, let us note that this definition requires some care as in general, we cannot guarantee that each $\omega_i$ is non degenerate (see e.g. the KdV example). Therefore, in the above definition we need to do two things. Viewing $\omega_i$ as a linear map from vertical vector fields to vertical $1$-forms, we restrict our attention to $0$-forms $f$ such that $\delta f$ is in the image of $\omega_i$. In other words, we consider $f$ such that there exists a (vertical) vector field ${\xi^i_f}$ which satisfies $\delta f=\ip{\xi^i_f} {\omega_i}$. In that case, we say that {\it $f$ is Hamiltonian with respect to $\omega_i$}. We also remedy the possible non trivial kernel by working modulo it, hence obtaining a non degenerate map, which we keep denoting $\omega_i$, on equivalence classes of vertical vector fields. This has no effect on the above definition of $\{f,g\}_i$ where $f$ and $g$ are two Hamiltonian $0$-forms with respect to $\omega_i$. We work with this understanding in the rest of the paper. 
\begin{thrm}[Decomposition of the multi-time Poisson Bracket]\label{decomposition}
Let $\displaystyle F= \sum_{i=1}^n F_i dx^i$ be a Hamiltonian $1$-form, then for $i=1,\dots,n$, $F_i$ is Hamiltonian with respect to $\omega_i$. Let $\displaystyle G= \sum_{i=1}^n G_i dx^i$ be another Hamiltonian $1$-form, then the following decomposition of the multi-time Poisson bracket holds:
\begin{equation}
\cpb{F}{G} = \sum_{i=1}^n \{F_i,G_i\}_i dx^i.
\end{equation}
\end{thrm}
\begin{proof}
On the one hand, by definition 
$$\delta F=\sum_{i=1}^n \delta F_i \wedge dx^i\,,$$
and on the other hand, since $F$ is Hamiltonian
\begin{equation}
\delta F=\ip{\xi_F}{\sum_{i=1}^n \omega_i \wedge dx^i}=\sum_{i=1}^n \ip{\xi_F}{\omega_i} \wedge dx^i\,, 
\end{equation}
hence $\delta F_i=\ip{\xi_F}{\omega_i}$ so $F_i$ is Hamiltonian with respect to $\omega_i$ for each $i=1,\dots,n$ and we can take $\xi^i_{F_i}=\xi_F$ for all $i=1,\dots,n$. Note that this gives an idea of how restrictive it is for $F$ to be Hamiltonian. 
Next, consider the following chain of equalities
\begin{eqnarray*}
\cpb{F}{G} &=& - \ip{\xi_F}{ \delta G}=- \ip{\xi_F}(\sum_{i=1}^n \delta G_i \wedge dx^i) = -\ip{\xi_F}(\sum_{i=1}^n \ip{\xi^i_{G_i}}\omega_i \wedge dx^i)\\
&=& \sum_{i=1}^n \ip{\xi^i_{G_i}} \ip{\xi_F} \omega_i \wedge dx^i=\sum_{i=1}^n \ip{\xi^i_{G_i}} \delta F_i \wedge dx^i=  \sum_{i=1}^n \{F_i,G_i\}_i dx^i
\end{eqnarray*}
which concludes the proof.  
\end{proof}
This is the generalization to an arbitrary number $n$ of flows in an integrable hierarchy of the decomposition theorem that was obtained in \cite{Caudrelier_Stoppato_2020} on examples. Theorem \ref{decomposition} provides a general proof, independent of examples, and reproduces the result of \cite{Caudrelier_Stoppato_2020} in the particular case $n=2$. This theorem describes the relationship between our multi-time Poisson bracket $\cpb{\;}{\;}$, encapsulating an arbitrary number of flows in the hierarchy, and the usual and dual single-time Poisson brackets $\pb{\;}{\;}_i$, which are related to each flow separately.

\section{Example: (potential) KdV hierarchy}\label{Section:pKdV}
In the following we will see the example of the KdV hierarchy with respect to its first two times, so in usual hierarchy notations, we would have $x_1=x$, $x_2=t_2$ and $x_3=t_3$ (if one consider the KdV alone, $t_3$ is simply the time $t$). In fact, we consider the potential form of the KdV hierarchy which is the appropriate form for a Lagrangian formulation. It is known that for KdV hierarchy the even flows are trivial $v_{2k} = 0$ $\forall k$, so we will also treat the less trivial case of the first two odd times $x_1=x$, $x_3=t_3$ and $x_5=t_5$. We use the Lagrangians multiforms presented in \cite{Vermeeren_2018}.
\subsection{Times 1,2 and 3}
\subsubsection{Multiform Euler-Lagrange equations}
We write the Hamiltonian formulation of the first two levels of the (potential) KdV hierarchy, described by the Lagrangian multiform $\Lag = L_{12}\, dx^{12} + L_{23}\, dx^{23} + L_{13}\, dx^{13}$, where
\begin{subequations}\label{kdv123:eq}
\begin{gather}
L_{12} = v_1 v_2,\\
L_{23} = - 3 v_1^2v_2 - v_1 v_{112} + v_{11} v_{12} - v_{111} v_2,\\
L_{13} = - 2 v_1^3 - v_1 v_{111} + v_1 v_3.
\end{gather}
\end{subequations}
One can easily check that the multiform E-L equations $\delta d \Lag = 0 $ are equivalent to
\begin{equation}
v_2 = 0, \qquad v_3 = v_{111} + 3v_1^2.
\end{equation}
and differential consequences: in particular we have the potential KdV from $v_{13} = (v_3)_1= v_{1111} + 6 v_1 v_{11}$.
\subsubsection{The symplectic multiform}
We are now going to show the procedure to obtain the symplectic multiform from $\Lag$ and \eqref{kdv123:eq}. We start by computing the $\delta$-differential of the Lagrangian multiform:
\begin{eqsplit}
\delta \Lag =&  v_1 \delta v_2 \wedge dx^{12} + v_2 \delta v_1 \wedge dx^{12} \\
&+(- 6 v_1 v_2 - v_{112}) \delta v_1 \wedge dx^{23} + (-3 v_1^2 - v_{111}) \delta v_2 \wedge dx^{23} + v_{12} \delta v_{11} \wedge dx^{23} \\
&+ v_{11} \delta v_{12} \wedge dx^{23} - v_1 \delta v_{112} \wedge dx^{23} - v_2 \delta v_{111} \wedge dx^{23}\\
&+(v_3 - v_{111} - 6 v_1^2) \delta v_1 \wedge dx^{13} + v_1 \delta v_3 \wedge dx^{13} - v_1 \delta v_{111} \wedge dx^{13}. 
\end{eqsplit} 
We now use the property $d\delta = - \delta d$ on some of the terms to obtain the desired expression $\delta \Lag = {\cal E}(\Lag) - d\omegone$, where ${\cal E}(\Lag)=0$ is equivalent to \eqref{kdv123:eq}. The reader can verify the following identities
\begin{equation}
 v_1 \delta v_2 \wedge dx^{12} = -v_{12}  \delta v \wedge dx^{12} - v_{13}  \delta v \wedge dx^{13}  - v_1 \delta v_3 \wedge dx^{13} - d (- v_1 \delta v \wedge dx^1),
\end{equation}
\begin{equation}
v_2 \delta v_1 \wedge dx^{12} = - v_{12} \delta v \wedge dx^{12} + v_{23} \delta v \wedge dx^{23} + v_2 \delta v_3 \wedge dx^{23} - d( v_2 \delta v \wedge dx^2),
\end{equation}
\begin{eqsplit}
(v_3 - v_{111} - 6 v_1^2) \delta v_1 \wedge dx^{13} =& - (v_3 - v_{111} - 6 v_1^2)_1 \delta v \wedge dx^{13} -(v_3 - v_{111} - 6 v_1^2)_2 \delta v \wedge dx^{23}\\ 
&- (v_3 - v_{111} - 6 v_1^2) \delta v_2 \wedge dx^{23} - d((v_3 - v_{111} - 6 v_1^2)\delta v \wedge dx^3),
\end{eqsplit}
\begin{eqsplit}
- v_1 \delta v_{111} \wedge dx^{13} =  &v_{1111} \delta v \wedge dx^{13} + v_{1112} \delta v \wedge dx^{23} + v_{111} \delta v_2 \wedge dx^{23}\\ 
&- v_{112} \delta v_1 \wedge dx^{23} - v_{11} \delta v_{12} \wedge dx^{23} + v_{12} \delta v_{11} \wedge dx^{23} + v_1 \delta v_{112}  \wedge dx^{23}\\ 
&- d( - v_1 \delta v_{11} \wedge dx^3 + v_{11} \delta v_1 \wedge dx^3 - v_{111} \delta v \wedge dx^3).
\end{eqsplit}
Using these identities in $\delta \Lag$ we get 
\begin{eqsplit}
\delta \Lag =& -2 v_{12} \delta v \wedge dx^{12} + (-2v_{13} + 2 v_{1111} + 12 v_1 v_{11}) \delta v \wedge dx^{13}\\
&+ (2v_{1112} + 12 v_1 v_{12}) \delta v \wedge dx^{23} +(- 6v_1 v_2 -2v_{112}) \delta v_1 \wedge dx^{23}\\
& +(-v_3 + v_{111} + 3 v_1^2) \delta v_2 \wedge dx^{23}+ v_2 \delta v_3 dx^{23} + 2v_{12} \delta v_{11} \wedge dx^{23} \\
&- v_2 \delta v_{111} \wedge dx^{23}\\
&- d\big( - v_1 \delta v \wedge dx^1 + v_2 \delta v \wedge dx^2 + (v_3 - 2 v_{111} - 6 v_1^2)\delta v \wedge dx^3\\
&+ v_{11} \delta v_1 \wedge dx^3 - v_1 \delta v_{11} \wedge dx^3 \big)\\
&\equiv {\cal E}(\Lag) - d\omegone
\end{eqsplit}
if we define $\omegone= - v_1 \delta v \wedge dx^1 + v_2 \delta v \wedge dx^2 + (v_3 - 2 v_{111} - 6 v_1^2)\delta v \wedge dx^3 + v_{11} \delta v_1 \wedge dx^3 - v_1 \delta v_{11} \wedge dx^3 $. We see that ${\cal E}(\Lag)= \delta \Lag + d \omegone=0$ is equivalent to the equations \eqref{kdv123:eq} and differential consequences. The symplectic multiform is then
\begin{eqsplit}
\Omega =& - \delta v_1 \wedge \delta v \wedge dx^1 + \delta v_2 \wedge \delta v \wedge dx^2 + \delta v_3 \wedge \delta v \wedge dx^3 \\
&- 2 \delta v_{111} \wedge \delta v \wedge dx^3 - 12 v_1 \delta v_1 \wedge \delta v \wedge dx^3 + 2 \delta v_{11} \wedge \delta v_1 \wedge dx^3.
\end{eqsplit}
\subsubsection{The Hamiltonian multiform}
We can now compute the Hamiltonian mutliform $\Ham = \displaystyle \sum_{i\le j} H_{ij}\, dx^{ij}$, using $$H_{ij}= \ip{\tilpartial_i}\omega^{(1)}_j -  \ip{\tilpartial_j}\omega^{(1)}_i - L_{ij}$$
    to find
\begin{subequations}
\begin{gather}
H_{12} = v_1 v_2,\\
H_{23} = - 3 v_1^2 v_2 - v_{111} v_2\\
H_{13}= v_1 v_3 - 4v_1^3 + v_{11}^2 - 2 v_1 v_{111}.
\end{gather}
\end{subequations}
The multiform Hamiltonian equations are obtained as
\begin{itemize}
\item $\delta H_{12} = \ip{\tilpartial_2}{\omega_1} - \ip{\tilpartial_1}{\omega_2}$:
\begin{equation}
v_1 \delta v_2 + v_2 \delta v_1 = - v_{12} \delta v + v_2 \delta v_1 - v_{12} \delta v + v_1 \delta v_2 \quad \implies \quad v_{12} = 0.
\end{equation}
\item $\delta H_{23} = \ip{\tilpartial_3}{\omega_2} - \ip{\tilpartial_2}{\omega_3}$:
\begin{eqsplit}
-3 v_1^2 \delta v_2 - 6 v_1 v_2 \delta v_1 - v_{111} \delta v_2 - v_2 \delta v_{111}  =& v_{23} \delta v - v_3 \delta v_2 - v_{23} \delta v + v_2 \delta v_3\\
& + 2 v_{1112} \delta v - 2 v_2 \delta v_{111} + 12 v_1 v_{12} \delta v\\
&- 12 v_1v_2 \delta v_1 - 2 v_{112} \delta v_1 + 2 v_{12} \delta v_{11} 
\end{eqsplit}
which implies the following system of equations
\begin{subequations}
\begin{gather}
v_2 = 0,\\
v_{12} = 0, \\
v_3 - 3v_1^2 - v_{111}=0,\\
v_{112} + 3 v_1 v_2 = 0,\\
v_{1112} + 6 v_1 v_{12} =0.
\end{gather}
\end{subequations}
\item $\delta H_{13} = \ip{\tilpartial_3}{\omega_1} - \ip{\tilpartial_1}{\omega_3}$:
\begin{eqsplit}
&v_1 \delta v_3 + v_3 \delta v_1 - 12 v_1^2 \delta v_1 + 2 v_{11} \delta v_{11} - 2 v_1 \delta v_{111} - 2 v_{111} \delta v_1 =\\
&- v_{13} \delta v + v_3 \delta v_1 - v_{13}\delta v + v_1 \delta v_3 + 2 v_{1111} \delta v\\
&- 2 v_1 \delta v_{111} + 12 v_1 v_{11} \delta v - 12 v_1^2 \delta v_1 - 2 v_{111} \delta v_1 + 2 v_{11} \delta v_{11},
\end{eqsplit}
which implies $v_{13} - v_{1111} - 6v_1 v_{11} = 0$.
\end{itemize}
This system of equations is equivalent to \eqref{kdv123:eq} as expected.
\subsubsection{Hamiltonian forms}
We now describe Hamiltonian forms for this case. A 1-form $Q=Q_1(v,v_1)\,dx^1 + Q_2(v,v_2)\, dx^2 + Q_3(v,v_1,v_3,v_{11},v_{111})\,dx^3$ for the symplectic multiform $\Omega$ is  Hamiltonian if and only if
\begin{equation}
\parder{Q_1}{v_1} = - \parder{Q_2}{v_2} = - \parder{Q_3}{v_3} = \frac{1}{2} \parder{Q_3}{v_{111}}, \qquad \parder{Q_1}{v} = \frac{1}{2} \parder{Q_3}{v_{11}}.
\end{equation}
Its related Hamiltonian vector field is
\begin{equation}
\xi_Q=\parder{Q_1}{v_1} \partial_v - \parder{Q_1}{v} \partial_{v_1} + \parder{Q_2}{v} \partial_{v_2} + \left(\parder{Q_3}{v} - 6 v_1 \parder{Q_3}{v_{11}}\right) \partial_{v_3} + \left(\frac{1}{2} \parder{Q_3}{v_1} - 3 v_1 \parder{Q_3}{v_{111}}\right) \partial_{v_{11}}.
\end{equation}
This can be proved as followed: one takes a generic vector field 
\begin{equation}
\xi_Q = A \partial_v + B \partial_{v_1} + C \partial_{v_2} + D \partial_{v_3} + E \partial_{v_{11}} + D \partial_{v_{111}}
\end{equation}
and determines the coefficients comparing the right and left hand-side of 
\begin{equation}
\ip{\xi_Q}{\Omega} = \delta Q. 
\end{equation}
This translates into constraints on the derivatives of $Q_i$ with respect to the field and its derivatives, and determines the coefficients of the vector field.\\
Here we verify that for any Hamiltonian 1-form $Q$ and modulo the equations of motion
\begin{equation}
d Q = \ip{\xi_Q}{\delta \Ham},
\end{equation}
or, more explicitly
\begin{itemize}
\item $\partial_1 Q_2 - \partial_2 Q_1=\ip{ \xi_Q}{\delta H_{12}}$, which means
\begin{eqsplit}
\parder{Q_2}{v} v_1 + \parder{Q_2}{v_2}v_{12} - \parder{Q_1}{v} v_2 - \parder{Q_1}{v_1}v_{12} =& - \parder{Q_1}{v} \parder{H_{12}}{v_1} + \parder{Q_2}{v} \parder{H_{12}}{v_2}\\
=&- \parder{Q_1}{v}v_2 + \parder{Q_2}{v} v_1 \\
&\implies \quad - 2  v_{12} \parder{Q_1}{v_1} = 0. 
\end{eqsplit} 
\item $\partial_2 Q_3 - \partial_3 Q_2=\ip{ \xi_Q}{\delta H_{23}}$, which means
\begin{eqsplit}
&\parder{Q_3}{v}v_2 + \parder{Q_3}{v_1}v_{12} + \parder{Q_3}{v_3} v_{23} + \parder{Q_3}{v_{11}}v_{112} + \parder{Q_3}{v_{111}} v_{1112} - \parder{Q_2}{v} v_3 - \parder{Q_2}{v_2}v_{23} \\
&= - \parder{Q_1}{v} \parder{H_{23}}{v_1} + \parder{Q_2}{v}\parder{H_{23}}{v_2}\\
&= 6 v_1 v_2 \parder{Q_1}{v} - (3 v_1^2 + v_{111}) \parder{Q_2}{v},
\end{eqsplit}
which again is
\begin{equation}
v_2 \parder{Q_3}{v} + v_{12} \parder{Q_3}{v_1} + 2v_{112} \parder{Q_1}{v_1} + (2v_{112} - 6 v_1 v_2) \parder{Q_1}{v} + (- v_3 + 3 v_1^2 + v_{111}) \parder{Q_2}{v} =0.
\end{equation}
\item  $\partial_1 Q_3 - \partial_3 Q_1=\ip{ \xi_Q}{\delta H_{13}}$, which means
\begin{eqsplit}
&\parder{Q_3}{v} v_1 + \parder{Q_3}{v_1} v_{11} + \parder{Q_3}{v_3} v_{13} + \parder{Q_3}{v_{11}} v_{111} + \parder{Q_3}{v_{111}} v_{1111} - \parder{Q_1}{v} v_3 - \parder{Q_1}{v_1} v_{13} \\
&=  - \parder{Q_1}{v} \parder{H_{13}}{v_1} + \left( \parder{Q_3}{v} - 6 v_1 \parder{Q_3}{v_{111}} \right) \parder{H_{13}}{v_3} + \left( \frac{1}{2} \parder{Q_3}{v_1} - 3 v_1 \parder{Q_3}{v_{111}} \right) \parder{H_{13}}{v_{11}}\\
&= (12 v_1^2 + 2 v_{111} - v_3) \parder{Q_1}{v} + v_1 \parder{Q_3}{v} - 6 v_1^2 \parder{Q_3}{v_{11}} + v_{11} \parder{Q_3}{v_1} - 6 v_1 v_{11} \parder{Q_3}{v_{111}},
\end{eqsplit}
which again is $(2v_{13} - 2v_{1111} - 12 v_{1} v_{11}) \parder{Q_3}{v_3}=0$. 
\end{itemize}
\subsubsection{Conservation Laws}
We can now find a conservation law for the Lagrangian multiform $\Lag$, i.e. a Hamiltonian 1-form $F = F_1(v,v_1) dx^1 + F_2(v, v_2) dx^2 + F_3(v,v_1,v_3,v_{11}, v_{111}) dx^3$ such that $\ip{\xi_F}{\delta \Ham} = \xi_F \Ham =0$:
\begin{itemize}
\item $\xi_F H_{12} = 0 $ means that $- \parder{F_1}{v}v_2 + \parder{F_2}{v} v_1=0$. Since $\parder{F_1}{v_1} = - \parder{F_2}{v_2}$, necessarily $F_1 = a(v) v_1 + b(v)$ and $F_2 = -a(v) v_2 + c(v)$. The condition above then translates to 
\begin{equation}
-a'(v) v_1v_2 -b'(v)v_2 -a'(v)v_1v_2 + c'(v)v_1=0 \quad \implies \quad a'(v) = b'(v) = c'(v) =0.
\end{equation}
We will set $a = 1$, and $b=c=0$, so we have $F_1 = v_1$ and $F_2 = -v_2$.
\item $\xi_F H_{23} = 6v_1 v_2 \parder{F_1}{v} - (3v_1^2 + v_{111}) \parder{F_2}{v} = 0$ automatically.
\item Because of the Hamiltonianity constraint we have that $F_3 = -v_3 +2 v_{111} + d(v,v_1)$. Now we solve for $d$ the equation $\xi_F H_{13} = (12 v_1^2 + 2 v_{111} - v_3) \parder{F_1}{v} + v_1 \parder{F_3}{v} - 6 v_1^2 \parder{F_3}{v_{11}} + v_{11} \parder{F_3}{v_1} - 6 v_1 v_{11} \parder{F_3}{v_{111}} = v_1 \parder{d(v,v_1)}{v} + v_{11} \parder{d(v,v_1)}{v_1} - 12 v_1 v_{11} =0$. This implies
\begin{equation}
\parder{d}{v} = 0, \quad \parder{d}{v_1} = 12 v_1 , \qquad \implies \qquad d= 6v_1^2.
\end{equation}
\end{itemize}
The conservation law is then
\begin{equation}
F = v_1 dx^1 - v_2 dx^2 + (-v_3 + 2v_{111} + 6v_1^2) dx^3.
\end{equation}
In fact its differential $dF$ is
\begin{eqsplit}
& v_{12} dx^{21} + v_{13} dx^{31} - v_{12} dx^{12} - v_{23} dx^{32} + (-v_{13} + 2v_{1111} + 12 v_1 v_{11}) dx^{13}\\
 &+ (-v_{23} + 2 v_{1112} + 12 v_1 v_{12}) dx^{23}=\\
&-2 v_{12} dx^{12} + (-2 v_{13} + 2v_{1111} + 12 v_1 v_{11}) dx^{13} + (2 v_{1112} + 12 v_1 v_{12}) dx^{23}
\end{eqsplit}
which vanishes on the equations of motion.
\subsubsection{Another symplectic multiform}
We now mention how to compute another symplectic multiform (and its related Hamiltonian multiform). One can perform an equivalent computation to the one above, making different choices as to what to apply $\delta d = - d \delta$ on, and obtain
\begin{equation}
\tilde \omegone =  - v_1 \delta v \wedge dx^1  + \frac{v_2}{2} \delta v \wedge dx^2 +\frac{1}{2} (v_3 - 9 v_1^2 - 3 v_{111}) \delta v\wedge dx^3 + v_{11} \delta v_1\wedge dx^3 - v_1 \delta v_{11}\wedge dx^3.
\end{equation}
It is easy to check that both $\delta \Lambda + d \tilde\omegone=0$ and $d(\omegone - \tilde \omegone)=0$ are  equivalent to \eqref{kdv123:eq}. We then define
\begin{eqsplit}
\tilde\Omega =& - \delta v_1 \wedge \delta v \wedge dx^1 + \frac{1}{2} \delta v_2 \wedge \delta v \wedge dx^2\\ 
&+ \frac{1}{2} \delta v_3 \wedge \delta v\wedge dx^3 - 9 v_1 \delta v_1 \wedge \delta v\wedge dx^3 - \frac{3}{2} \delta v_{111} \wedge \delta v\wedge dx^3 + 2 \delta v_{11} \wedge \delta v_1\wedge dx^3.
\end{eqsplit}
The coefficients of Hamiltonian multiform $\tilde{\Ham} = \tilde H_{12}\, dx^{12} + \tilde H_{23}\, dx^{23} + \tilde H_{13} \,dx^{13}$ are
\begin{subequations}
\begin{gather}
\tilde H_{12} = \frac{1}{2}v_1 v_2, \\
\tilde H_{23} = - \frac{3}{2} v_1^2 v_2 - \frac{1}{2}v_2 v_{111},\\
\tilde H_{13}= \frac{1}{2}v_1 v_3 + v_{11}^2 - \frac{5}{2} v_1^3 - \frac{5}{2}v_1 v_{111}
\end{gather} 
\end{subequations}
and the multiform Hamilton equations for $\tilde {\Ham}$ and $\tilde \Omega$ bring the same set of equations as expected.
\subsection{Times 1,3 and 5}
\subsubsection{The symplectic and Hamiltonian multiform}
In the previous section we considered the times 1 2 and 3 of (potential) KdV hierarchy. We can also describe the odd-time flows 1, 3 and 5, using the Lagrangian multiform $\Lag = L_{13}\, dx^{13} + L_{15}\, dx^{15} + L_{35}\, dx^{35}$, where
\begin{subequations}
\begin{gather}
L_{13} = -2 v_1^3 + v_1 v_3 - v_1 v_{111},\\
L_{15} = - 5 v_1^4 + 10 v_1 v_{11}^2 + v_1 v_5 - v_{111}^2,
\end{gather}
\begin{eqsplit}
L_{35} =& 6 v_1^5 - 10 v_1^3 v_3 + 20 v_1^3 v_{111} - 15 v_1^2 v_{11}^2 + 3 v_1^2 v_5 + 3 v_1^2 v_{11111} -10 v_1 v_3 v_{111}\\
&+ 20 v_1 v_{11} v_{13} - 12 v_1 v_{11} v_{1111} + 6 v_1 v_{111}^2 - 5 v_3 v_{11}^2 + 7 v_{11}^2 v_{111} + v_1 v_{115} \\
&- v_3 v_{11111} + v_5 v_{111} - v_{11} v_{15} + 2 v_{13} v_{1111} - 2 v_{111} v_{113} + v_{111} v_{11111} - v_{1111}^2.
\end{eqsplit}
\end{subequations}
The multiform E-L equations are equivalent to 
\begin{equation}\label{kdv135:eq}
v_3 = v_{111} + 3 v_1^2 , \qquad v_5 = v_{11111} + 10 v_1^3 + 5 v_{11}^2 + 10 v_1 v_{111}.
\end{equation}
and differential consequences. If we define the form $\omegone$ to be
\begin{eqsplit}
\omegone =& - v_1 \delta v \wedge dx^1 + (v_3 - 2 v_{111} - 6v_1^2) \delta v \wedge dx^3 + v_{11} \delta v_1 \wedge dx^3 - v_1 \delta v_{11} \wedge dx^3\\
&+(v_5 - 20 v_1^3 - 20 v_1 v_{111} - 10 v_{11}^2 - 2 v_{11111}) \delta v \wedge dx^5 \\
&+ (20 v_1 v_{11} + 2 v_{1111}) \delta v_1 \wedge dx^5 - 2 v_{111} \delta v_{11} \wedge dx^5,
\end{eqsplit}
one can check that $\delta \Lambda + d \omegone = 0$ are equivalent to \eqref{kdv135:eq}. The symplectic multiform is then $\Omega = \omega_1 \wedge dx^1 + \omega_3 \wedge dx^3 + \omega_5 \wedge dx^5$, where
\begin{subequations}
\begin{gather}
\omega_1 = \delta v \wedge \delta v_1,\\
\omega_3 = \delta v_3 \wedge \delta v -2 \delta v_{111} \wedge \delta v +2 \delta v_{11} \wedge \delta v_1 - 12 v_1 \delta v_1 \wedge \delta v,
\end{gather}
\begin{eqsplit}
\omega_5 =& \delta v_5 \wedge \delta v + (60 v_1^2 + 20 v_{111} ) \delta v \wedge \delta v_1 - 20 v_1 \delta v_{111} \wedge \delta v - 20 v_{11} \delta v_{11} \wedge \delta v\\
&- 2 \delta v_{11111} \wedge \delta v + 20 v_1 \delta v_{11} \wedge \delta v_1 + 2 \delta v_{1111} \wedge \delta v_1 - 2 \delta v_{111} \wedge \delta v_{11}.
\end{eqsplit}
\end{subequations}
The Hamiltonian multiform is obtained in the usual way and reads $\Ham = H_{13}\,dx^{13} + H_{35}\, dx^{35} + H_{15}\, dx^{15}$ where
\begin{subequations}
\begin{gather}
H_{13} = v_1 v_3 + v_{11}^2 - 2 v_1 v_{111} - 4 v_1^3,\\
H_{15} = v_1 v_5 - 15 v_1^4 - 20 v_1^2 v_{111} - 2 v_{11111} v_1 + 2 v_{1111} v_1 - v_{111}^2,
\end{gather}
\begin{eqsplit}
H_{35} = &- 10 v_1^3 v_3 - 10 v_1 v_{111} v_3 - 5 v_{11}^2 v_3 - v_{11111} v_3 + v_{111} v_5 +3 v_1^2 v_5 - 6 v_1^5 - 20 v_1^3v_{111} \\
&+ 15 v_1^2 v_{11}^2 - 3 v_1^2 v_{11111} + 12 v_1 v_{11} v_{1111} - 6 v_1 v_{111}^2 - 7 v_{11}^2 v_{111} - v_{111} v_{11111} + v_{1111}^2.
\end{eqsplit}
\end{subequations}
One can then proceed in a similar way to the 123-times case and verify the validity of the multiform Hamilton equations.\\
\subsubsection{Hamiltonian forms and conservation laws}
We obtain that a $1$-form 
\begin{equation}
F = F_1(v,v_1) \, dx^1 + F_3 (v,v_1,v_3,v_{11},v_{111}) \, dx^3+ F_5(v, v_1,v_5, v_{11},v_{111},v_{1111},v_{11111}) \, dx^5
\end{equation}
is Hamiltonian if and only if
\begin{subequations}\label{kdv135:hamform}
\begin{gather}
    \parder{F_1}{v_1} = \frac{1}{2} \parder{F_3}{v_{111}} = \frac{1}{2} \parder{F_5}{v_{11111}} = - \parder{F_3}{v_3} = - \parder{F_5}{v_5},\\
    \parder{F_1}{v} = \frac{1}{2} \parder{F_3}{v_{11}} =\frac{1}{2} \parder{F_5}{v_{1111}}, \qquad \parder{F_5}{v_{111}} = \parder{F_3}{v_1} + 4v_1 \parder{F_5}{v_{11111}}.
\end{gather}
\end{subequations}
Its related Hamiltonian vector field is
\begin{eqsplit}
    \xi_F =& \parder{F_1}{v_1} \partial_v - \parder{F_1}{v}\partial_{v_1} + \left(\parder{F_3}{v} + + 10 v_{11} \parder{F_3}{v_{111}} + 4 v_1 \parder{F_3}{v_{11}} - \parder{F_5}{v_{11}} \right) \partial_{v_3}\\
    &+ \left( \parder{F_5}{v} - 10 v_1 \parder{F_5}{v_{11}} + 10 v_{11} \parder{F_5}{v_{111}} + (70 v_1^2 - 10 v_{111}) \parder{F_5}{v_{1111}} \right) \partial_{v_5} \\
    &+ \left(\frac{1}{2} \parder{F_3}{v_1} - 3 v_1 \parder{F_3}{v_{111}}\right) \partial_{v_{11}} + \left(- \frac{1}{2} \parder{F_5}{v_{11}} + 5 v_{11} \parder{F_5}{v_{11111}} + 5 v_1 \parder{F_5}{v_{1111}} \right) \partial_{v_{111}} \\
    & + \left( \frac{1}{2} \parder{F_5}{v_1} - 5 v_1 \parder{F_5}{v_{111}} + (35 v_1^2 - 5 v_{111})\parder{F_5}{v_{11111}}\right) \partial_{v_{1111}}\,.
\end{eqsplit}
From the equations \eqref{kdv135:hamform} one can obtain a conservation law:
\begin{equation}
    F= v_1 \, dx^1 + (-v_3 + 2 v_{111} + 6 v_1^2) \, dx^3 + (-v_5 + 2 v_{11111} + 20 v_1 v_{111} + 5 v_{11}^2 + 10 v_1^2) \, dx^5\,.
\end{equation}
\section{Example: sine-Gordon hierarchy in light-cone coordinates}\label{Section:SG}

In this section we will show another example, i.e. the first two levels of the sine-Gordon hierarchy in light-cone coordinates. 

\subsection{Multiform Euler-Lagrange equations}
A lagrangian multiform for this set of equations has been obtained for the first time in \cite{Suris_2016} and is $\Lag=L_{12}\,dx^{12} + L_{13}\, dx^{13}+ L_{23}\,dx^{23}$, where
\begin{subequations}
\begin{gather}
L_{12} = \frac{1}{2} u_1 u_2 - \cos u,\\
L_{13} = \frac{1}{2} u_1 u_3 + \frac{1}{2}u^2_{11}-\frac{1}{8}u_1^4,\\
L_{23} = - \frac{1}{2} u_2 u_3 + u_{11} u_{12} - u_{11}\sin u + \frac{1}{2} u_1^2 \cos{u}.
\end{gather}
\end{subequations}
The multiform E-L equations $d \delta \Lag=0$ are equivalent to
\begin{equation}\label{SG123:eq}
u_{12} - \sin u=0,\qquad u_{3} - \frac{1}{2}u^3_1  - u_{111}=0
\end{equation}
and differential consequences.

\subsection{The symplectic and Hamiltonian multiform}
An similar computation to the ones above yields the form $\omegone$ as
\begin{equation}
 \omegone= - \frac{1}{2} u_1 \delta u \wedge dx^1 + \frac{1}{2} u_2 \delta u \wedge dx^2  - (\frac{u_{111}}{2} + \frac{u_{1}^3}{4}) \delta u \wedge dx^3 + u_{11} \delta u_1 \wedge dx^3.
\end{equation}
The $\delta$-differential of $\omegone$ is the symplectic multiform $\Omega  =   \omega_1 \wedge dx^1+\omega_2 \wedge dx^2+\omega_3 \wedge dx^3$, with
\begin{subequations}
\begin{gather}
\omega_1 = \frac{1}{2} \delta u \wedge \delta u_1,\\
\omega_2 = \frac{1}{2} \delta u_2 \wedge \delta u, \\
\omega_3 = - \frac{1}{2} \delta u_{111} \wedge \delta u - \frac{3u_1^2}{4} \delta u_1 \wedge \delta u + \delta u_{11}\wedge \delta u_1 .
\end{gather} 
\end{subequations}
The Hamiltonian multiform $\Ham = H_{12}\, dx^{12}  + H_{13}\, dx^{13}  + H_{23} \,dx^{23}$ is computed as 
\begin{subequations}
\begin{gather}
H_{12} = \frac{1}{2} u_1 u_2 + \cos u \,,\\
H_{13} = -\frac{1}{2} u_1 u_{111} + \frac{1}{2} u_{11}^2 - \frac{1}{8} u_1^4 \,,\\
H_{23} =- \frac{1}{2} u_2 u_{11} - \frac{1}{4} u_1^2 u_2 + u_{11} \sin u - \frac{1}{2} u_1^2 \cos u  \,.
\end{gather}
\end{subequations}
The multiform Hamilton equations are obtained as $\displaystyle\delta \mathcal{H} = \sum_{j=1}^3 dx^j \wedge \tilpartial_j \lrcorner \Omega$ and are equivalent to the generalised Euler-Lagrange equations, as required.
\subsection{Hamiltonian forms and multi-time Poisson brackets}
One can then investigate the presence of Hamiltonian forms:
\begin{itemize}
\item A 0-form $H(u, u_1,u_2,u_{11},u_{111})$ is always Hamiltonian, with Hamiltonian vector field
\begin{eqsplit}
\xi_H=& \left(2 \parder{H}{u_1} - 3 u_1^2 \parder{H}{u_{111}} \right) \partial_u \wedge \partial_1 - 2 \parder{H}{u_2} \partial_u \wedge \partial_2\\ 
&+ 2 \parder{H}{u_{111}}\partial_u \wedge \partial_3 +\left(\frac{3}{2} u_1^2 \parder{H}{u_{11}} - 2 \parder{H}{u} \right) \partial_{u_1} \wedge \partial_1 - \parder{H}{u_{11}} \partial_{u_1} \wedge \partial_3.
\end{eqsplit}
We remark that $\xi_H$ is not unique;
\item A 1-form $P= P_1 \,dx^1 + P_2 \,dx^2 + P_3 \,dx^3$  is Hamiltonian if $P_1 = P_1 (u,u_1)$, $P_2=P_2(u,u_2)$, $P_3=P_3(u,u_1,u_{11},u_{111})$, and  
\begin{gather}
\parder{P_3}{u_{11}} = 2 \parder{P_1}{u}, \quad \parder{P_3}{u_{111}} = - \parder{P_2}{u_2} = \parder{P_1}{u_{1}}.
\end{gather}
and its related vector field is
\begin{eqsplit}
\xi_P =& 2 \parder{P_1}{u_1} \partial_u - 2 \parder{P_1}{u} \partial_{u_1} +2 \parder{P_2}{u} \partial_{u_2} \\
&+  \left( \parder{P_3}{u_1} - \frac{3}{2} u_1^2 \parder{P_3}{u_{111}} \right)\partial_{u_{11}} + \left(\frac{3}{2} u_1^2 \parder{P_3}{u_{11}} - 2 \parder{P_3}{u}\right)\partial_{u_{111}} .
\end{eqsplit}
\item The only Hamiltonian 2-forms or 3-forms are the constant ones.
\end{itemize}
For such forms we can define the multi-time Poisson brackets. The Poisson bracket between an Hamiltonian 0-form $H$ and an Hamiltonian 1-form $P=P_1\, dx^1 + P_2\, dx^2 + P_3 \,dx^3$ is $\xi_P H$, therefore
\begin{eqsplit}
\cpb{H}{P} =& 2 \parder{P_1}{u_1}\parder{H}{u} - 2 \parder{P_1}{u} \parder{H}{u_1} + 2 \parder{P_2}{u} \parder{H}{u_2} - 2 \parder{P_3}{u}\parder{H}{u_{111}}\\
&+ \parder{P_3}{u_1} \parder{H}{u_{11}} - \frac{3}{2} u_1^2 \parder{P_3}{u_{111}}\parder{H}{u_{11}} + \frac{3}{2}u_1^2 \parder{P_3}{u_{11}}\parder{H}{u_{111}} 
\end{eqsplit}
If $\displaystyle P = \sum_{i=1}^3 P_i dx^i$ and $\displaystyle Q= \sum_{i=1}^3 Q_i dx^i$ are Hamiltonian 1-forms, then their Poisson bracket satisfies the decomposition
\begin{equation}
\cpb{P}{Q} = \{P_1,Q_1\}_1\, dx^1 +\{P_2,Q_2\}_2\, dx^2 + \{P_3,Q_3\}_3\, dx^3,
\end{equation}
where 
\begin{subequations}
\begin{gather}
\pb{P_1}{Q_1}_1 = 2 \parder{P_1}{u} \parder{Q_1}{u_1} - 2\parder{P_1}{u_1} \parder{Q_1}{u} \,,\\
\pb{P_2}{Q_2}_2 = 2  \parder{P_2}{u_2}\parder{Q_2}{u} -2  \parder{P_2}{u}\parder{Q_2}{u_2}  \,,
\end{gather}
\begin{eqsplit}
\pb{P_3}{Q_3}_3 =&2 \parder{P_3}{u} \parder{Q_3}{u_{111}} - 2 \parder{P_3}{u_{111}}\parder{Q_3}{u}   + \parder{P_3}{u_{11}} \parder{Q_3}{u_1} - \parder{P_3}{u_1} \parder{Q_3}{u_{11}} \\
&+ \frac{3}{2} u_1^2 \parder{P_3}{u_{111}}\parder{Q_3}{u_{11}} -   \frac{3}{2} u_1^2 \parder{P_3}{u_{11}}\parder{Q_3}{u_{111}}  \,.
\end{eqsplit}
\end{subequations}

Contrary to the pKdV example, (and AKNS example below), for the sG we were not able to find a Hamiltonian $1$-form producing conservation laws in the sense of Definition \ref{cons_law}.
However, it is possible to find a non-Hamiltonian 1-form $F= F_1\, dx^1 + F_2\, dx^2 + F_3\, dx^3$ that is closed on the equations of motion, as follows:
\begin{equation}
F_1=\frac{1}{2}u_1^2\,,~~F_2=-\cos u\,,~~F_3=\frac{3}{8}u_1^4+u_1u_{111}-\frac{1}{2}u_{11}^2\,.
\end{equation}
Then, on the equations of motion, one checks that 
\begin{eqnarray}
\partial_1F_2=\partial_2F_1\,,~~\partial_1F_3=\partial_3F_1\,,~~\partial_2F_3=\partial_3 F_2\,.
\end{eqnarray}
Thus, the sG example points to a need to extend our approach to conservation laws beyond Hamiltonian forms.

\section{Example: AKNS hierarchy}\label{Section:AKNS}

Our last example deals with the AKNS hierarchy. For this example, we include one more time compared to previous example, to remind the reader that in principle we can keep adding more times in a multiform, corresponding to adding more and more flows in the hierarchy. As becomes clear in this example, the explicit expression soon become cumbersome though.

\subsection{Multiform Euler-Lagrange equations}
We start from the Lagrangian multiform found in \cite{Sleigh_Nijhoff_Caudrelier_2019_2} 
$$\Lag = L_{12} \,dx^{12} + L_{13}\, dx^{13} + L_{14} \, dx^{14} + L_{23}\, dx^{23} + L_{24} \, dx^{24} + L_{34} \, dx^{34}\,,$$ 
where
\begin{subequations}
	\begin{gather}
	L_{12} = \frac{1}{2}(rq_2- qr_2) + \frac{i}{2}q_1 r_1 + \frac{i}{2} q^2 r^2,\\
	L_{13}= \frac{1}{2}(r q_3-q r_3  ) - \frac{1}{8}(r_1 q_{11} - q_1 r_{11}) - \frac{3qr}{8} (r q_1 -q r_1),\\
	L_{14} = \frac{1}{2} (r q_4 -  q r_4) - \frac{5i}{16} qr (q r_{11} + r q_{11}) - \frac{3i}{16} (q^2 r_1^2 + q_1^2 r^2) - \frac{i}{4} qr q_1 r_1 + \frac{i}{8} q_{11} r_{11}  + \frac{i}{4} q^3 r^3,
	\end{gather}
	\begin{eqsplit}
		L_{23} = & \frac{1}{4}(q_2 r_{11} - r_2 q_{11}) - \frac{i}{2}(q_3 r_1 + r_3 q_1) + \frac{1}{8}(q_1 r_{12} - r_1 q_{12})\\ 
		&+ \frac{3qr}{8}(qr_2- r q_2) - \frac{i}{8}q_{11}r_{11} + \frac{i}{4} q r(q r_{11} + r q_{11}) - \frac{i}{8} (q r_1 - r q_1)^2 - \frac{i}{2} q^3 r^3.
	\end{eqsplit}
	\begin{eqsplit}
		L_{24} = & \frac{3}{8} q^2 r^2 (r q_1 - q r_1) - \frac{i}{16} (q^2 r_1 r_2 + r^2 q_1 q_2) - \frac{5i}{16} qr ( q r_{12} + r q_{12}) - \frac{1}{8} qr(r q_{111} - q r_{111})\\ 
		&- \frac{1}{8}(q^2 r_1 r_{11} - r^2 q_1 q_{11}) - \frac{1}{8}q_1 r_1(r q_1 - q r_1) + \frac{1}{4} qr (r_1 q_{11} - q_1 r_{11}) + \frac{3i}{8} qr (q_1 r_2 + r_1 q_2)\\ 
		&- \frac{i}{8} (q_{111} r_2 + r_{111} q_2) + \frac{1}{16} (q_{111} r_{11} - r_{111} q_{11}) + \frac{i}{8} (q_{11} r_{12} + r_{11} q_{12}) - \frac{i}{2} (q_1 r_4 + r_1 q_4),
	\end{eqsplit}
	\begin{eqsplit}
		L_{34} =& \frac{i}{8} (q_{11} r_{13} + r_{11} q_{13}) - \frac{i}{8} (q_{111} r_3 + r_{111} q_3) - \frac{i}{32} q_{111} r_{111} + \frac{i}{32} (q^2 r_{11}^2 + r^2 q_{11}^2) \\
		&+ \frac{i}{32} q_1^2 r_1^2 + \frac{3}{8} qr(r q_4 - q r_4) + \frac{9i}{32} q^4 r^4 - \frac{3i}{16} q^2 r^2 (q r_{11} + r q_{11}) - \frac{i}{16} (q^2 r_1 r_3 + r^2 q_1 q_3)\\ 
		&- \frac{5i }{16}qr (q r_{13} + r q_{13}) + \frac{1}{4} (q_{11} r_4 - r_{11} q_4) + \frac{3i}{16}qr (q_1 r_{111} + r_1 q_{111}) + \frac{i}{16} qr q_{11} r_{11} \\
		&- \frac{i}{16} q_1 r_1 (q r_{11} + r q_{11}) - \frac{15i}{16} q^2 r^2 q_1 r_1 + \frac{3i}{8} qr (q_1 r_3 + r_1 q_3) - \frac{1}{8} (q_1 r_{14} - r_1 q_{14}),
	\end{eqsplit}
\end{subequations}
The corresponding multiform Euler-Lagrange equations $\delta d L=0$ produce the familiar first three levels of the AKNS hierarchy
\begin{subequations}
\begin{gather}
\label{AKNS1}
q_2 - \frac{i}{2} q_{11} + i q^2r=0, \qquad r_2 + \frac{i}{2} r_{11} - i q r^2=0,\\
q_3 + \frac{1}{4} q_{111} - \frac{3}{2} qr q_1=0, \qquad r_3 + \frac{1}{4}r_{111} - \frac{3}{2} qrr_1=0,\\
q_4 =- \frac{i}{8} q_{1111} - \frac{3i}{4} q^3 r^2 + \frac{i}{4} q^2 r_{11} + \frac{i}{2} q q_1 r_1  + i q r q_{11} + \frac{3i}{4} q_1^2 r ,\\
\label{AKNS4}
r_4 = \frac{i}{8} r_{1111} +  \frac{3i}{4} q^2 r^3 - \frac{i}{4} r^2 q_{11} - \frac{i}{2} r q_1 r_1  - i q r r_{11} - \frac{3i}{4} r_1^2 q.
\end{gather}
\end{subequations}

\subsection{The symplectic and Hamiltonian multiforms}
As done in the previous two examples, the computation of the form $\omegone$ from $\delta {\cal L}$ gives
\begin{eqsplit}
\omegone =& \left(-\frac{1}{2} r \delta q  + \frac{1}{2} q \delta r \right)\wedge dx^1 + \left(\frac{i}{2} q_1 \delta r + \frac{i}{2} r_1 \delta q\right)\wedge dx^2 \\
&+ \left(\left(\frac{1}{4} r_{11} - \frac{3}{8}qr^2\right) \delta q   +\left( - \frac{1}{4} q_{11} + \frac{3}{8}q^2 r\right) \delta r  - \frac{1}{8} r_1 \delta q_1  + \frac{1}{8}q_1 \delta r_1\right)\wedge dx^3\\
&+\bigg(\left( - \frac{i}{8} r_{111} - \frac{i}{16} q_1 r^2 + \frac{3i}{8} qr r_1\right) \delta q  +\left( - \frac{i}{8} q_{111} - \frac{i}{16} q^2 r_1 + \frac{3i}{8} qr q_1\right) \delta r \\
&+ \left( \frac{i}{8} r_{11} - \frac{5i}{16} qr^2 \right) \delta q_1  +  \left( \frac{i}{8} q_{11} - \frac{5i}{16} q^2r \right) \delta r_1 \bigg) \wedge dx^4 .
\end{eqsplit}
and its $\delta$-differential is the symplectic multiform
\begin{equation}
\Omega = \omega_1 \wedge dx^1 + \omega_2 \wedge dx^2 + \omega_3 \wedge dx^3 + \omega_4 \wedge dx^4,
\end{equation}
where 
\begin{subequations}
\begin{gather}
\omega_1 = \delta q \wedge \delta r\,,\\
\omega_2 = \frac{i}{2}(\delta q_1 \wedge \delta r + \delta r_1 \wedge \delta q)\,,\\
\omega_3 = \frac{1}{4}( \delta r_{11} \wedge \delta q - \delta q_{11} \wedge \delta r) + \frac{1}{4} \delta q_1 \wedge \delta r_1 + \frac{3qr}{2} \delta q \wedge \delta r\,,
\end{gather}
\begin{eqsplit}
\omega_4 =& - \frac{i}{8} \delta r_{111} \wedge \delta q - \frac{i}{8} \delta q_{111} \wedge \delta r+ \frac{i}{4}r^2 \delta q_1 \wedge \delta q - \frac{i}{4} q^2 \delta r_1 \wedge \delta r \\
&+ i qr \delta q_1 \wedge \delta r + i qr \delta r_1 \wedge \delta q + \frac{i}{2} ( q_1 r - q r_1) \delta q \wedge \delta r 
\end{eqsplit}
\end{subequations}
The Hamiltonian multiform $\Ham = H_{12}\, dx^{12}  + H_{13}\, dx^{13}+ H_{14} \, dx^{14} + H_{23}\, dx^{23} + H_{24} \, dx^{24} + H_{34}\, dx^{34}$ can now be computed and brings
\begin{subequations}
\begin{gather}
H_{12} = \frac{i}{2}(q_1 r_1 - q^2 r^2)\,,\\
H_{13} = \frac{1}{4}(r_{11} q_1-q_{11} r_1)\,\\
H_{14} = \frac{i}{8} (q_1^2 r^2 + q^2 r_1^2) - \frac{i}{8} (q_{111} r_1 + q_1 r_{111}) + \frac{i}{8} q_{11} r_{11} + i qr q_1 r_1 - \frac{i}{4} q^3 r^3, \\
H_{23} = \frac{i}{8} q_{11} r_{11} -\frac{iqr}{4} (q r_{11} + q_{11} r) + \frac{i}{8}(q r_1- q_1r)^2 + \frac{i}{2} q^3 r^3\,.
\end{gather}
\begin{eqsplit}
H_{24} =& \frac{qr}{8}(r q_{111} - q r_{111}) + \frac{3}{8} q^2 r^2 (qr_1-r q_1 ) + \frac{q_1 r_1}{8}(rq_1 - q r_1)\\
&+ \frac{qr}{4} (q_1 r_{11}-r_1 q_{11}) + \frac{1}{16} (q_{11} r_{111} - q_{111} r_{11}) + \frac{1}{8} (q^2 r_1 r_{11} - r^2 q_1 q_{11}), 
\end{eqsplit}  
\begin{eqsplit}
H_{34} =& \frac{i}{16} (q^2 r_1 r_3 + r^2 q_1 q_3) + \frac{i}{32} q_{111} r_{111} - \frac{i}{32} (q^2 r_{11}^2 + q_{11}^2 r^2) - \frac{i}{32} q_1^2 r_1^2\\
&- \frac{9i}{32} q^4 r^4 + \frac{3i}{16} q^2 r^2 (q r_{11} + q_{11} r) - \frac{3i}{16} qr (q_1 r_{111} + r_1 q_{111}) \\
&- \frac{i}{16} qr q_{11} r_{11} + \frac{i}{16} q_1r_1 (q r_{11} + q_{11} r) + \frac{15i}{16} q^2 r^2 q_1 r_1.
\end{eqsplit}
\end{subequations}
The multiform Hamilton equations are obtained as $\displaystyle\delta \mathcal{H} = \sum_{j=1}^4 dx^j \wedge \tilpartial_j \lrcorner \Omega$. One checks with a direct computation that they indeed reproduce the set of equations \eqref{AKNS1}-\eqref{AKNS4}. We remark that $H_{12}$ and $H_{13}$ are the covariant Hamiltonian densities of respectively the NLS equations and the modified KdV equation already obtained in \cite{Caudrelier_Stoppato_2020}.

\subsection{Hamiltonian forms and multi-time Poisson brackets}
We have the following facts
\begin{itemize}
\item Any 0-form $H$ is Hamiltonian;
\item A 1-form 
\begin{eqsplit}
    F =& F_1(q,r)\, dx^1 + F_2(q,r,q_1,r_1)\, dx^2 + F_3(q,r,q_1,r_1,q_{11},r_{11})\, dx^3\\
    &+ F_4(q,r,q_1,r_1,q_{11},r_{11},q_{111},r_{111}) \,dx^4
\end{eqsplit}
is Hamiltonian if
\begin{eqsplit}
\parder{F_1}{r} = 2i \parder{F_2}{r_1} = -4 \parder{F_3}{r_{11}} = -8i\parder{F_4}{r_{111}}\,,  &\qquad\parder{F_1}{q} = -2i \parder{F_2}{q_1} = -4 \parder{F_3}{q_{11}} = 8i\parder{F_4}{q_{111}}\,,\\
\parder{F_2}{r} = 2i \parder{F_3}{r_1} = - 4 \parder{F_4}{r_{11}}\,,  &\qquad \parder{F_2}{q} = -2i \parder{F_3}{q_1} = - 4 \parder{F_4}{q_{11}}\,,\\
\parder{F_4}{r_1} = -\frac{i}{4}qr \parder{F_1}{r} - \frac{i}{2} \parder{F_3}{r} + \frac{i}{4}q^2 \parder{F_1}{q}\,, & \qquad \parder{F_4}{q_1} = \frac{i}{4}qr \parder{F_1}{q} + \frac{i}{2} \parder{F_3}{q} - \frac{i}{4}r^2 \parder{F_1}{r}\,.
\end{eqsplit} 
and its related vector field is
\begin{eqsplit}
\xi_F =& \parder{F_1}{r} \partial_q - \parder{F_1}{q} \partial_r -2i \parder{F_2}{r} \partial_{q_1}-2i \parder{F_2}{q} \partial_{r_1} \\
&-4 \left(6qr \parder{F_3}{r_{11}} + \parder{F_3}{r}\right)\partial_{q_{11}} +4\left(6qr \parder{F_3}{q_{11}} + \parder{F_3}{q}\right)\partial_{r_{11}}\\
&+8i\left(\parder{F_4}{r} + 2q^2 \parder{F_4}{q_{11}} + 8qr\parder{F_4}{r_{11}} + 4(q_1r-r_1q) \parder{F_4}{r_{111}} \right)\partial_{q_{111}}\\
&+8i\left(\parder{F_4}{q} + 2r^2 \parder{F_4}{r_{11}} + 8qr\parder{F_4}{q_{11}} - 4(q_1r-r_1q) \parder{F_4}{q_{111}} \right) \partial_{r_{111}}\, .
\end{eqsplit}
\end{itemize}
We can write the general expression of a Hamiltonian 1-form, given the first coefficient $F_1(q,r)$.  Since $\parder{F_1}{r} = 2i \parder{F_2}{r_1}$ and $\parder{F_1}{q} = -2i \parder{F_2}{q_1}$ we need
\begin{equation}
F_2 = \frac{i}{2} \left(\parder{F_1}{q} q_1 - \parder{F_1}{r} r_1\right) + a(q,r).
\end{equation}
$a(q,r)$ is a term left to determine. Then, since $\parder{F_3}{q_{11}} = -\frac{1}{4} \parder{F_1}{q}$ and $\parder{F3}{r_{11}} = - \frac{1}{4} \parder{F_1}{r}$ we have 
\begin{equation}
F_3 = - \frac{1}{4} \parder{F_1}{q} q_{11} - \frac{1}{4} \parder{F_1}{r} r_{11} + (\dots)(q,r,q_1,r_1).
\end{equation}
Then we use the fact that $\parder{F_2}{r} = 2i \parder{F_3}{r_1}$ and $\parder{F_2}{q} = -2i \parder{F_3}{q_1}$ to obtain
\begin{equation}
\parder{F_3}{q_1} = \frac{1}{4} \left( \parder{^2 F_1}{q \partial r} r_1 - \parder{^2 F_1}{q^2} q_1 \right) + \frac{i}{2} \parder{a}{q}, \qquad \parder{F_3}{r_1} = \frac{1}{4}\left(  \parder{^2 F_1}{q \partial r} q_1 - \parder{^2 F_1}{r^2} r_1\right)- \frac{i}{2} \parder{a}{r}.
\end{equation}
we then use partial integration and find
\begin{equation}
F_3 = - \frac{1}{4} \parder{F_1}{q} q_{11} - \frac{1}{4} \parder{F_1}{r} r_{11}- \frac{1}{8} \left(\parder{^2 F_1}{r^2} r_1^2 + \parder{^2F_1}{q^2}q_1^2 - 2 \parder{^2 F_1}{q \partial r}q_1 r_1 \right)+ \frac{i}{2} \parder{a}{q} q_1 - \frac{i}{2} \parder{a}{r} r_1 + b(q,r),
\end{equation}
where $b(q,r)$ is another term left to determine. Similarly we can compute the fourth coefficient $F_4$, which is
\begin{eqsplit}
F_4 =& \frac{i}{8}\left(\parder{F_1}{r} r_{111} - \parder{F_1}{q} q_{111} + \parder{^2 F_1}{r^2} r_1 r_{11} - \parder{^2 F_1}{q^2} q_1 q_{11} + \parder{^2F_1}{q \partial r} (r_1 q_{11} - q_1 r_{11}) \right)\\
&- \frac{1}{4} \left( \parder{a}{r} r_{11} + \parder{a}{q} q_{11} \right) + \frac{i}{4}\left(qrq_1 +q^2 r_1\right) \parder{F_1}{q} - \frac{i}{4} \left(qrr_1 + r^2 q_1\right)\parder{F_1}{r} \\
&+ \frac{i}{48} \left( \parder{^3 F_1}{r^3} r_1^3 - \parder{^3 F_1}{q^3}q_1^3\right) + \frac{i}{16} \left( \parder{^3 F_1}{q^2 \partial r} q_1^2 r_1 -\parder{^3 F_1}{q \partial r^2} q_1 r_1^2  \right)\\
&- \frac{1}{8}\left( \parder{^2 a}{q^2} q_1^2 + \parder{^2 a}{r^2} r_1^2 + 2 \parder{^2 a}{q \partial r} q_1 r_1 \right) + \frac{i}{2} \left(\parder{b}{q}q_1 - \parder{b}{r} r_1 \right) + c(q,r)
\end{eqsplit}
where $c(q,r)$ is left to determine.\\

For Hamiltonian forms we can define the multi-time Poisson brackets. The Poisson bracket between a 0-form $H(q,r,q_1,r_1,q_{11},r_{11}, q_{111}, r_{111})$ and an Hamiltonian 1-form $P=P_1\, dx^1 + P_2 \,dx^2 + P_3 \,dx^3 + P_4\, dx^4$ is $\xi_P H$:
\begin{eqsplit}
\cpb{H}{P}& = \parder{H}{q} \parder{P_1}{r} -  \parder{H}{r} \parder{P_1}{q} - 2i \parder{H}{r_1} \parder{P_2}{q} - 2i \parder{H}{q_1}\parder{P_2}{r}\\ 
&+ 24qr \parder{P_3}{q_{11}} \parder{H}{r_{11}} - 24 qr \parder{P_3}{r_{11}}\parder{H}{q_{11}} + 4 \parder{H}{r_{11}}\parder{P_3}{q} - 4 \parder{H}{q_{11}}\parder{P_3}{r}\\
&+8i\left(\parder{P_4}{r} + 2q^2 \parder{P_4}{q_{11}} + 8qr\parder{P_4}{r_{11}} + 4(q_1r-r_1q) \parder{P_4}{r_{111}} \right)\parder{H}{q_{111}}\\
&+8i\left(\parder{P_4}{q} + 2r^2 \parder{P_4}{r_{11}} + 8qr\parder{P_4}{q_{11}} - 4(q_1r-r_1q) \parder{P_4}{q_{111}} \right) \parder{H}{r_{111}}\,.
\end{eqsplit}
If $\displaystyle P = \sum_{i=1}^4 P_i dx^i$ and $\displaystyle Q= \sum_{i=1}^4 Q_i dx^i$ are Hamiltonian 1-forms, then their Poisson bracket satisfies the decomposition
\begin{equation}
\cpb{P}{Q} = -\{P_1,Q_1\}_1\, dx^1 +\{P_2,Q_2\}_2\, dx^2 + \{P_3,Q_3\}_3\, dx^3 + \{P_4,Q_4\}_4 \, dx^4,
\end{equation}
where 
\begin{subequations}
\begin{gather}
\pb{P_1}{Q_1}_1 = \parder{P_1}{r}\parder{Q_1}{q} - \parder{P_1}{q}\parder{Q_1}{r}\,,\\
\pb{P_2}{Q_2}_2 = 2i \left( \parder{P_2}{q}\parder{Q_2}{r_1} - \parder{P_2}{r_1}\parder{Q_2}{q} + \parder{P_2}{r}\parder{Q_2}{q_1} - \parder{P_2}{q_1}\parder{Q_2}{r} \right) \,,
\end{gather}
\begin{eqsplit}
\pb{P_3}{Q_3}_3 =& 4\bigg( \parder{P_3}{r_{11}} \parder{Q_3}{q} -\parder{P_3}{q} \parder{Q_3}{r_{11}} - \parder{P_3}{q_{11}} \parder{Q_3}{r} + \parder{P_3}{r} \parder{Q_3}{q_{11}} \\
&+ \parder{P_3}{q_1} \parder{Q_3}{r_1} - \parder{P_3}{r_1} \parder{Q_3}{q_1} -6qr( \parder{P_3}{q_{11}} \parder{Q_3}{r_{11}} - \parder{P_3}{r_{11}} \parder{Q_3}{q_{11}}) \bigg)\,,
\end{eqsplit}
\begin{eqsplit}
\pb{P_4}{Q_4}_4 =&- 8i \left( \parder{P_4}{q_{111}} \parder{Q_4}{r} -  \parder{P_4}{r} \parder{Q_4}{q_{111}}\right) + 8i  \left( \parder{P_4}{q} \parder{Q_4}{r_{111}} -  \parder{P_4}{r_{111}} \parder{Q_4}{q}\right)\\
&-  8i  \left( \parder{P_4}{q_1} \parder{Q_4}{r_{11}} -  \parder{P_4}{r_{11}} \parder{Q_4}{q_1}\right) +  8i\left( \parder{P_4}{q_{11}} \parder{Q_4}{r_{1}} -  \parder{P_4}{r_{1}} \parder{Q_4}{q_{11}}\right)\\
&+64iqr \left( \parder{P_4}{q_{11}} \parder{Q_4}{r_{111}} -  \parder{P_4}{r_{111}} \parder{Q_4}{q_{11}}\right) -64iqr \left( \parder{P_4}{q_{111}} \parder{Q_4}{r_{11}} -  \parder{P_4}{r_{11}} \parder{Q_4}{q_{111}}\right)\\
&+ 16iq^2 \left( \parder{P_4}{q_{11}} \parder{Q_4}{q_{111}} - \parder{P_4}{q_{111}} \parder{Q_4}{q_{11}}\right)+  16ir^2 \left( \parder{P_4}{r_{11}} \parder{Q_4}{r_{111}} - \parder{P_4}{r_{111}} \parder{Q_4}{r_{11}}\right) \\
&+ 32 i(q_1 r - r_1 q) \left( \parder{P_4}{r_{111}} \parder{Q_4}{q_{111}} - \parder{P_4}{q_{111}} \parder{Q_4}{r_{111}} \right)
\end{eqsplit}
\end{subequations}
Using this we can read the single-time Poisson brackets: $\pb{\;}{\;}_1$ is the usual equal-time Poisson bracket of the AKNS hierarchy, which in the traditional infinite dimensional setting provide the first structure (in the sense of bi-Hamiltonian theory) for the whole hierarchy, while $\pb{\;}{\;}_{2,3}$ are the dual Poisson Bracket of respectively the NLS and mKdV which can be found in \cite{Avan_Caudrelier_Doikou_Kundu_2016}. We remark the presence of a minus sign in front of $dx^1$, which is just so $\pb{\;}{\;}_1$ reproduces the usual single-time Poisson Brackets with the right sign. It remains an open question to relate our findings with the traditional theory of bi-Hamiltonian structures \`a la Magri \cite{Magri_1978}. 

\subsection{Conservation laws}
Since the coefficients of the Hamiltonian multiform are Hamiltonian, the multiform Hamilton equations in a Poisson bracket form are
\begin{equation}
d F = \ip{\xi_F}{\delta H} =\sum_{i<j =1}^4 \cpb{H_{ij}}{F} \,dx^{ij}
\end{equation} 
for any Hamiltonian 1-form $F = F_1\, dx^1 + F_2\, dx^2 + F_3\, dx^3 + F_4 \, dx^4$. We can also find the first conservation laws for the AKNS hierarchy, i.e. $F$ is a conservation law if and only if 
\begin{equation}
\cpb{H_{ij}}{F} =0 \qquad \forall i<j\,.
\end{equation} 
We can solve the latter equation in the space of Hamiltonian forms (see previous section for the general expression of the cofficients) to find a conservation law. From $(i,j)= (1,2)$ we get 
\begin{equation}
    \cpb{H_{12}}{F} = - i qr^2 \parder{F_1}{r} + i q^2 r \parder{F_1}{q} + \frac{i}{2} \parder{^2F_1}{q^2}q_1^2 - \frac{i}{2} \parder{^2F_1}{r^2}r_1^2 + \parder{a}{q}q_1 + \parder{a}{r} r_1 = 0\,.
\end{equation}
This translates into $r \parder{F_1}{r} = q \parder{F_1}{q}$ and $\parder{^2F_1}{q^2} = \parder{^2 F_1}{r^2} = 0$, and therefore $F_1 = qr$, and $\parder{a}{q} = \parder{a}{r} =0$, so therefore $a$ is constant, which we set to zero. The coefficients become then
\begin{subequations}
\begin{equation}
    F_1 = qr, \quad F_2 = \frac{i}{2} \left(rq_1 - qr_1\right), \quad F_3 = -\frac{1}{4}r q_{11} - \frac{1}{4} q r_{11} + \frac{1}{4}q_1 r_1 + b(q,r)\,,
\end{equation}
\begin{equation}
    F_4 = \frac{i}{8} \left( q r_{111} - r q_{111}  + q_{11} r_1 - q_1 r_{11}\right) + \frac{i}{2} \left( \parder{b}{q} q_1 - \parder{b}{r} r_1 \right) + c(q,r)\,
\end{equation}
\end{subequations}
with $b$ and $c$ left to determine. From $(i,j) = (1,3)$ we get 
\begin{equation}
    \cpb{H_{13}}{F} = - \frac{3}{2} (qr^2q_1 + q^2rr_1) + q_1 \parder{b}{q} + r_1 \parder{b}{r} =0\,,
\end{equation}
and therefore $b= \frac{3}{4}q^2 r^2$. The fourth coefficient becomes then $F_4 = \frac{i}{8} ( q r_{111} - r q_{111}) + \frac{i}{8} (q_{11} r_1 - q_1 r_{11}) +  \frac{3i}{4} qr(q_1 r - qr_1) +  c(q,r) $. It can be verified by looking at the coefficient $(1,4)$ that we have a conservation law when $c=0$. The conservation law is then
\begin{eqsplit}
F = &qr\, dx^1 + \frac{i}{2}(q_1 r- r_1 q)\, dx^2 + \frac{1}{4}(3q^2 r^2 + q_1 r_1 -q_{11}r - r_{11}q)\,dx^3\\
&+ \left(\frac{i}{8} ( q r_{111} - r q_{111}) + \frac{i}{8} (q_{11} r_1 - q_1 r_{11}) +  \frac{3i}{4} qr(q_1 r - qr_1) \right) \, dx^4.
\end{eqsplit}

\section*{Acknowledgements}
It is a pleasure to acknowledge helpful discussions with Frank Nijhoff, Duncan Sleigh, and Mats Vermeeren.

\section*{Data Availability Statement}
Data sharing is not applicable to this article as no new data were created or analyzed in this study.

\bibliography{bibliography_stoppato} 
\bibliographystyle{ieeetr}
\end{document}